\documentclass[letterpaper,12pt]{amsart}
\usepackage{setspace}
\usepackage{latexsym,array,delarray,amsthm,amssymb,epsfig,color}
\usepackage{graphics,graphicx, caption, comment, subcaption}
\usepackage{longtable,tikz}
\usepackage{todonotes}
\usepackage{enumerate}
\usepackage{hyperref}
\usepackage{hhline}

\usetikzlibrary{calc,intersections,through,backgrounds}

\pgfdeclarelayer{background}
\pgfsetlayers{background,main}

\theoremstyle{plain}
\newtheorem{thm}{Theorem}[section]
\newtheorem{lemma}[thm]{Lemma}
\newtheorem{prop}[thm]{Proposition}
\newtheorem{cor}[thm]{Corollary}
\newtheorem{conj}[thm]{Conjecture}
\newtheorem*{thm*}{Theorem \ref{thm:main}}
\newtheorem*{lemma*}{Lemma}
\newtheorem*{prop*}{Proposition}
\newtheorem*{cor*}{Corollary}
\newtheorem*{conj*}{Conjecture}

\theoremstyle{definition}
\newtheorem{defn}[thm]{Definition}
\newtheorem{ex}[thm]{Example}

\newtheorem{alg}[thm]{Algorithm}

\theoremstyle{remark}
\newtheorem*{rmk}{Remark}

\newcommand{\rr}{\mathbb{R}}
\newcommand{\R}{\mathbb{R}}

\newcommand{\ff}{\mathbb F}

\newcommand{\calc}{\mathcal{C}}

\newcommand{\LL}{\mathcal{L}}

\newcommand{\C}{\mathcal C}
\newcommand{\U}{\mathcal U}
\newcommand{\D}{\mathcal{D}}
\newcommand{\Z}{\mathcal{Z}}
\newcommand{\X}{\mathcal{X}}
\newcommand{\Y}{\mathcal{Y}}

\newcommand{\ind}{\mbox{$\perp \kern-5.5pt \perp$}}

\DeclareMathOperator*{\supp}{supp}

\tikzstyle{vertex} = [fill,shape=circle,node distance=80pt]
\tikzstyle{edge} = [fill,opacity=.5,fill opacity=.5,line cap=round, line join=round, line width=50pt]
\tikzstyle{elabel} =  [fill,shape=circle,node distance=30pt]

\newcommand{\E}{\mathcal{E}}
\newcommand{\Gonered}{\Gamma_{1_{red}}}
\newcommand{\Goneblue}{\Gamma_{1_{blue}}}
\newcommand{\Gtwored}{\Gamma_{2_{red}}}
\newcommand{\Gtwoblue}{\Gamma_{2_{blue}}}

\begin{document}

\title[Neural ideals and stimulus space visualization]{Neural ideals and stimulus space visualization}
\author{Elizabeth Gross}
\author{Nida Kazi Obatake}
\author{Nora Youngs}

             \email{elizabeth.gross@sjsu.edu}
              \email{nida.kazi@sjsu.edu} 
              \email{nyoungs@hmc.edu}
               
              \address{Department of Mathematics and Statistics, One Washington Square,  San Jos\'{e} State University, San Jos\'{e}, CA, 95192-0103, USA}
              \address{Department of Mathematics, Harvey Mudd College, 301 Platt Boulevard, Claremont, CA, 91711-5901}
              
           


\maketitle

\begin{abstract} 

A neural code $\C$ is a collection of binary vectors of a given length n that record the co-firing patterns of a set of neurons. Our focus is on neural codes arising from place cells, neurons that respond to geographic stimulus. In this setting, the stimulus space can be visualized as subset of $\R^2$ covered by a collection $\mathcal U$ of convex sets such that the arrangement $\mathcal U$ forms an Euler diagram for $\C$.   There are some methods to determine whether such a convex realization $\mathcal U$ exists; however, these methods do not describe how to draw a realization. In this work, we look at the problem of algorithmically drawing Euler diagrams for neural codes using two polynomial ideals: the neural ideal, a pseudo-monomial ideal; and the neural toric ideal, a binomial ideal.  In particular, we study how these objects are related to the theory of piercings in information visualization, and we show how minimal generating sets of the ideals reveal whether or not a code is $0$, $1$, or $2$-inductively pierced.
\end{abstract}


\section{Introduction}

In 2014, the Nobel Prize in Medicine or Physiology was awarded to John O'Keefe and  his team for their 1971 discovery of place cells \cite{OD71}.  A {\it place cell} is a neuron that codes a distinct region in an animal's environment called a {\it place field}. That is, if the animal is in a place field, the associated place cell fires; otherwise it is silent. Such neurons are believed to be an essential part of the navigation system and spatial memory.

The firing activity of a population of neurons over time results in a set of co-firing patterns, which can be stored using binary vectors, or \emph{codewords}.  Each codeword indicates the set of neurons that were firing together during some time window.  A set $\C\subset\{0,1\}^n$ of codewords on $n$ neurons is called a {\it combinatorial neural code}; the descriptor ``combinatorial" is commonly used since the precise details of neural spiking and timing are discarded, leaving only discrete co-firing patterns.  For a description of how neuronal firing data may be discretized, see \cite{CI08}. 
  
Each codeword in a combinatorial neural code $\C$ is associated with the set of neurons it represents; that is, given $c\in \C\subset\{0,1\}^n$, we associate $c$ with $Z_c:=\supp(c) = \{i\in [n] \, | \, c_i =1\}$.  If the neurons in question are known to be place cells, then a codeword $c$ of co-firing place cells indicates that the neurons in $Z_c$ have overlapping place fields.

Place fields can be approximated by convex sets in $\R^2$, for example, see \cite[Figure 1]{MRC15}.  Given an arrangement of convex subsets of $\R^2$ representing place fields, we can easily extract the associated neural code by considering the various zones in the arrangement.  That is, given a collection of sets $\U = \{U_1,...,U_n\}$ with each $U_i\subset\R^2$ a convex set, the code associated to $\U$ is  $$\C(\U) = \{c\in \{0,1\}^n \, | \, \big(\bigcap_{i\in Z_c} U_i\big)\backslash \big(\bigcup_{j\notin Z_c} U_j\big)\neq \emptyset\},$$ as illustrated in Figure \ref{fig:codecorr}.  We define $\bigcap_{i\in \emptyset} U_i = \R^2$ and refer to $U_i$ as the place field of neuron $i$.

\begin{figure}[h] 
   \centering
   \includegraphics[width=2in]{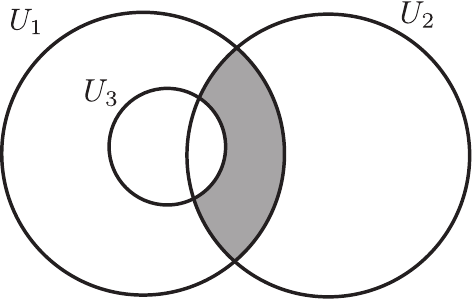} 
   \caption{An arrangement of three sets $\U = \{U_1,U_2,U_3\}$; the codeword associated to the shaded region is $c=110$ and $Z_c=\supp(c) = \{1,2\}$.  Here, the associated code is $\C(\U) = \left\{000,100, 010, 110, 101,111\right\}$.}
   \label{fig:codecorr}
\end{figure}

 The inverse problem is more difficult: given a particular neural code $\C$ presumed to come from place cells, can we find a set of convex subsets in $\R^2$ which would, as place fields,  exhibit $\C$ as its associated code?  If such a collection of convex sets exists, the code is called \emph{convexly realizable in $\R^2$}.  Previous work \cite{neural_ring, MRC15} has considered the question of determining whether or not a neural code is convexly realizable in $\R^2$ from the viewpoint of convex geometry and algebraic topology.  Once it is determined that a code is convexly realizable, however, it is not yet known how to algorithmically construct a realization.  

Since a neural code can be viewed as a set of relationships between $n$ sets, realizations of neural codes are \emph{Euler diagrams}, which have been studied since the 1700's \cite{Ham1860}; Venn diagrams are examples of Euler diagrams.  Thus the motivating question at the center of this work is: \emph{Given a neural code, how does one draw its corresponding Euler diagram?} While Euler diagrams using convex sets can occasionally be drawn by hand with some creativity, drawing Euler diagrams using convex sets automatically has been challenging. However, techniques have been proposed in the field of information visualization \cite{FH02}, \cite{Cho07}, \cite{RZF08}, \cite{SAA09}, including one particularly efficient method using the theory of piercings \cite{SZHR11}.  Specifically, in \cite{SZHR11} Stapleton \emph{et al.}  give a polynomial time algorithm to draw a realization using circles if the code is \emph{inductively pierced}.  We give a precise definition of inductively pierced in Section \ref{sect:defnsandnotation}, however we note here that $0$, $1$, and $2$-inductively pierced codes are simple to intuit. For example, a code is  $1$-inductively pierced if there exists a realization using circles that can be drawn iteratively, where at each step a closed curve is added with the condition that the new curve can intersect at most one previously drawn curve.  

The existence of such a drawing algorithm as the one in \cite{SZHR11} means that experimentalists can produce stimulus space visualizations of inductively pierced codes.  Thus, in this paper, we focus on connecting the theory of piercings to computational algebraic geometry and developing techniques for checking whether a code is $k$-inductively pierced.  This follows the footsteps of algebraic statistics \cite{PS05, algStatBook}, which has used tools from commutative algebra and algebraic geometry to look at problems in computational biology for the past decade and a half. We take an ideal theoretic focus, examining two ideals, the \emph{neural ideal}, a pseudo-monomial ideal, and the \emph{toric ideal of a neural code}, a binomial ideal.

The neural ideal and its canonical form were first introduced in \cite{neural_ring}, and have been used to answer questions regarding the place field structure of a neural code, including set relationships \cite{neural_ring} and convexity \cite{MRC15}.  For a code $\C \subset \{0,1\}^n$,  the {\it neural ideal} $J_\C$ is defined as follows: 
$$J_\C = \langle \prod_{v_i = 1} x_i \prod_{v_j=0}(1-x_j) \ | \ v\in \{0,1\}^n\backslash \C\rangle \subseteq \ff_2[x_1,...,x_n]$$  

In \cite{neural_ring}, the authors show that the pseudo-monomials in the neural ideal correspond to set containments in realizations of the code:

\begin{prop}\label{prop:RFstructure} \cite[Lemma 4.2]{neural_ring} If $ \{U_1,...,U_n\}$ is a collection of sets with corresponding code $\C$, then the pseudo-monomial $\prod_{i\in \sigma} x_i \prod_{j\in \tau} (1-x_j)$ is in $J_\C$ if and only if $\left(\bigcap_{i\in \sigma} U_i \right)\subseteq \bigcup_{j\in \tau} U_j$.
\end{prop}

This relationship between the neural ideal and place field structures lends itself naturally to answering questions regarding automatically drawing realizations.  In Section \ref{sect:canonicalform}, we make this connection concrete by giving necessary conditions on the canonical form for $k$-inductively pierced codes, and necessary and sufficient conditions on the canonical form for $0$- and $1$-inductively pierced codes. Once it is  determined that a code is inductively pierced, then the algorithm for automatically drawing Euler diagrams developed by Stapleton \emph{et al.} \cite{SZHR11} may be applied to draw a place field diagram for the neural code.

In addition to the neural ideal and its canonical form, there are other algebraic objects that can aid in place field visualization; in this manuscript, we introduce \emph{toric ideals of neural codes}.  In general, toric ideals are binomial ideals which have been well studied due to their underlying combinatorial structure \cite{MS05} and their expansive list of applications, including categorical data analysis \cite{DS98}, network modeling \cite{PRF}, evolutionary biology \cite{SturmSulli}, systems biology \cite{CDSS09}, integer programming \cite{GB}, geometric modeling \cite{CGS08}, and mathematical physics \cite{ABW14}.  
 
A toric ideal is most commonly defined in terms of an integer matrix, however, for this application, we will define it in terms of the neural code.  The connection to the standard definition of a toric ideal is straightforward, since we can treat each codeword as a column in an $n \times m$ matrix. Let $\mathcal C=\{c_1, \ldots, c_m\}$ be a neural code on $n$ neurons and define $\phi_\C$ as follows:
\begin{align*}
\phi_{\mathcal C}:\mathbb{K}[p_c\ |\ c\in\calc \setminus \{ 00\ldots00 \}] & \longrightarrow\mathbb{K}[x_i\ |\ i\in\{1,\ldots,n\}] \\  
p_c & \longmapsto\prod\limits_{i\in\supp(c)}^{}x_i.
\end{align*}  
The \emph{toric ideal of $\calc$}, denoted $I_\calc$, is the kernel of the map $\phi_{\mathcal C}$.  
 
The advantage of working with toric ideals is that there is software available (e.g. {\tt 4ti2} \cite{4ti2}) for working with these particular type of ideals that is interfaced with both {\tt Sage} \cite{Ste16} and {\tt Macaulay2} \cite{M2}.  Thus performing computations is straightforward and fast when $n$ is reasonable.  For example, finding the generators for a toric ideal of a code on 50 neurons consisting of 75 codewords took 0.059 seconds on a 2015 MacBook Air with a 2.2 GHz Intel Core i7 processor. 

In Section \ref{sect:toricideal}, we give degree bounds on the generators of $I_{\mathcal C}$ when $\mathcal C$ is $0$ or $1$-inductively pierced.  In fact, the toric ideals associated to $1$-inductively pierced codes form a class of toric ideals generated by quadratics, and thus these ideals would be interesting to study from a purely combinatorial commutative algebra viewpoint. 

In order to prove the main theorems from Section \ref{sect:toricideal}, we rely on the fact that toric ideals arising from 0-1 matrices are \emph{toric ideals of hypergraphs}, which have been studied in \cite{Vill}, \cite{PS14}, and \cite{GP13}.  In fact, we use the machinery for establishing degree bounds in \cite{GP13} to give the degree bound in Theorem \ref{thm:1piercings}.

The paper is organized as follows.  Section 2 focuses on the basic definitions and notation that provide the foundation for this paper.  In Section 2, we introduce Euler diagrams, $k$-piercings, and formally define what it means for a code to be $k$-inductively pierced.  We also introduce the neural ideal and its canonical form and toric ideals of neural codes.  We close Section 2 by reviewing the needed theory on toric ideals of hypergraphs. In Section 3,  we show how to detect $k$-piercings from the canonical form of the neural ideal, describe the canonical form structure of $k$-inductively pierced codes, and give necessary and sufficient conditions for $0$- and $1$-inductively pierced codes using the neural ideal.  We end Section 3 by giving an algorithm for finding piercing orderings for $1$-inductively pierced codes using the canonical form.  In Section 4, we show that a neural code is $0$-inductively pierced if and only if its toric ideal is the zero ideal.  We then show that the toric ideal of a $1$-inductively pierced code is generated by quadratics and give preliminary evidence for a stronger conjecture regarding $1$-inductively pierced codes.  In Section 5, we conclude by working through an example of a neural code with 17 neurons and 28 codewords using the methods described in Sections 3 and 4.





\section{Definitions and Notation}\label{sect:defnsandnotation}

\subsection{Euler diagrams and $k$-inductively pierced codes}
Let $\C\subset\{0,1\}^n$ be a code on $n$ neurons.  We will assume $\C$ contains the all-zeros word and all neurons fire at least once, that is, for each $i\in [n]$, there is at least one codeword $c\in \C$ with $c_i = 1$. In other words, we are assuming all place fields are non-empty.  

An \emph{Euler diagram} $d$ for $n$ sets is a collection of $n$ labeled simple, closed  curves ($\lambda_1$, $\lambda_2$, \ldots, $\lambda_n$) in $\rr^2$.  The interior of the curve $\lambda_i$ is a subset $U_i$ of $\rr^2$, i.e. $U_i = \text{int } \lambda_i$.  Denoting the boundary of $U_i$ as $\partial U_i$, we have that $\lambda_i = \partial U_i$. Non-empty intersections of the sets $U_1, \ldots, U_n$ and their complements $\bar U_1, \ldots, \bar U_n$ form regions called \emph{zones}; specifically, a set $Z$ is a \emph{zone} in diagram $d$ if $(\bigcap_{i\in Z} U_i )\cap (\bigcap_{j\notin Z} \bar U_j )$ is nonempty. An Euler diagram is said to be \emph{well-formed} \cite{SZHR11} if it satisfies the following conditions:
\begin{enumerate}
\item Each curve label is used only once.
\item All curves intersect generally (so curves intersect in finitely many points.)
\item A point in the plane is passed through at most 2 times by the curves in the diagram.
\item Each zone is connected. 
\end{enumerate}

Because we will focus on well-formed Euler diagrams in this manuscript and well-formedness requires each curve label to be used only once, we will use $\lambda_i$ to denote both the $i$th curve and the label of the $i$th curve.  We call a diagram \emph{convex} if each $U_i=\text{int} \lambda_i$ is convex. 

An \emph{abstract description} $\D=(\LL, \Z)$ of an Euler diagram $d$ is an ordered pair specifying the curve labels $\LL$ and the zones  $\Z \subseteq \mathcal{P}(\LL)$, where $\mathcal{P}(\LL)$ denotes the power set of $\LL$. We will assume $\emptyset \in \Z$ and if $\lambda \in \LL$, then there exists a $Z \in \Z$ such that $\lambda \in Z$. We will call an Euler diagram $d$ with abstract description $\D$ a \emph{realization} or \emph{drawing} of $\D$.

 Let $c \in \{0, 1\}^n$ and $Z_c = \text{supp c} \subseteq [n]$.  A neural code $\mathcal C$ on $n$ neurons corresponds naturally to the abstract description $\D_{\mathcal C} = ( [n], \Z_{\mathcal C})$ where $\Z_{\mathcal C} = \{ Z_c \ : c \in \mathcal C\}$. We call an abstract description $\D$ \emph{well-formed} if there exists a well-formed realization of $\D$, and we call a code $\C$ \emph{well-formed} if $\D_{\C}$ is well-formed.

We now describe two subsets of the power set $\mathcal{P}(\LL)$ that will be used in the definition of a $k$-piercing of an abstract description. Let $\D=(\LL, \Z)$ be an abstract description. Given $\lambda \in \LL$,  let $\X_{\lambda} \subset \Z$ be the set of all zones that contain $\lambda$:
$$ \X_{\lambda} = \{ Z \in \Z \ : \ \lambda \in Z \}.$$
Given $Z \in \Z$ and $\Lambda \subset \LL$ such that $Z \cap \Lambda = \emptyset$, let the \emph{$\Lambda$-cluster of $Z$}, denoted $\Y_{Z, \Lambda}$ be the set:
$$\Y_{Z, \Lambda} \ = \ \{ Z \cup \Lambda_i \ : \ \Lambda_i \subseteq \Lambda \}. $$

\begin{defn}\label{def:k-piercing} \cite{SZHR11}  Let $\D = (\LL,\Z)$ be an abstract description.  Let $\Lambda=\{ \lambda_1,...,\lambda_{k} \} \subseteq \LL$ be distinct curve labels.  Then $\lambda_{k+1} \in L$ is a \emph{k-piercing} of $\Lambda$ in $\D$ if there exists a zone $Z$ such that 
\begin{enumerate}
\item $\lambda_i\notin Z$ for each $i\leq k+1$
\item $\X_{\lambda_{k+1}}  = \Y_{ Z \cup \{\lambda_{k+1} \},  \Lambda}$, and 
\item $\Y_{Z, \Lambda} \subseteq \Z$.
\end{enumerate}
When the above conditions hold, we say that $\lambda_{k+1}$ is a $k$-piercing {\it identified} by the background zone $Z$.
\end{defn}

As we will focus primarily on 0- and 1-piercings in this paper, we now give a more detailed description and example of each type.
 
\begin{ex}[0-piercings]\label{def:0-piercing}
A $0$-piercing is a curve that intersects zero other curves in the diagram.   Let $\D=(\LL,\Z)$ be an abstract description.  Then $\lambda$ is a $0$-piercing in $\D$ if there exists a zone $z$ such that $\lambda \notin Z$, $\X_{\lambda}  = \{ Z \cup \{\lambda\} \}$, and $Z \cup \{\lambda \} \in \Z$.



As an example of a $0$-piercing, let $\C =\{000,100,010,101\}$.  Then \newline $\D_{\mathcal C} = \Big\{\big\{1,2,3\big\},\big\{\emptyset,\{1\},\{2\},\{1,3\}\big\}\Big\}$, and the label ``3" is a 0-piercing in $\D_{\mathcal C}$ identified by the zone $Z_{100}=\{1\}$.  Figure \ref{fig:0-piercingexample} illustrates the 0-piercing in terms of a place field diagram of $\C$. 

\begin{figure}[!ht] 
  \begin{subfigure}[b]{0.5\linewidth}
    \centering
    \includegraphics[width=0.7\linewidth]{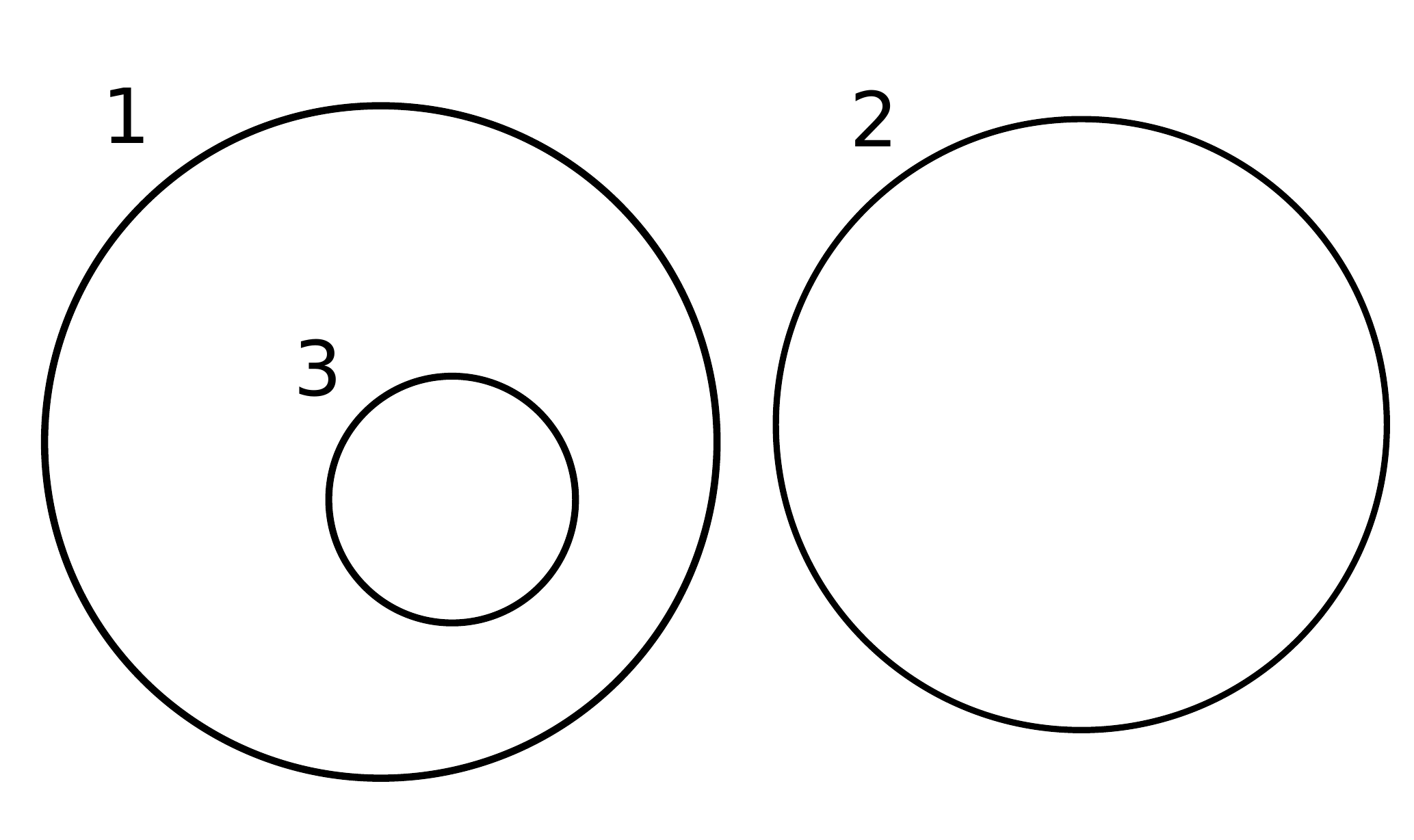} 
    \caption{A diagram with a 0-piercing.}
 \label{fig:0-piercingexample}
  \end{subfigure}
  \begin{subfigure}[b]{0.5\linewidth}
    \centering
    \includegraphics[width=0.7\linewidth]{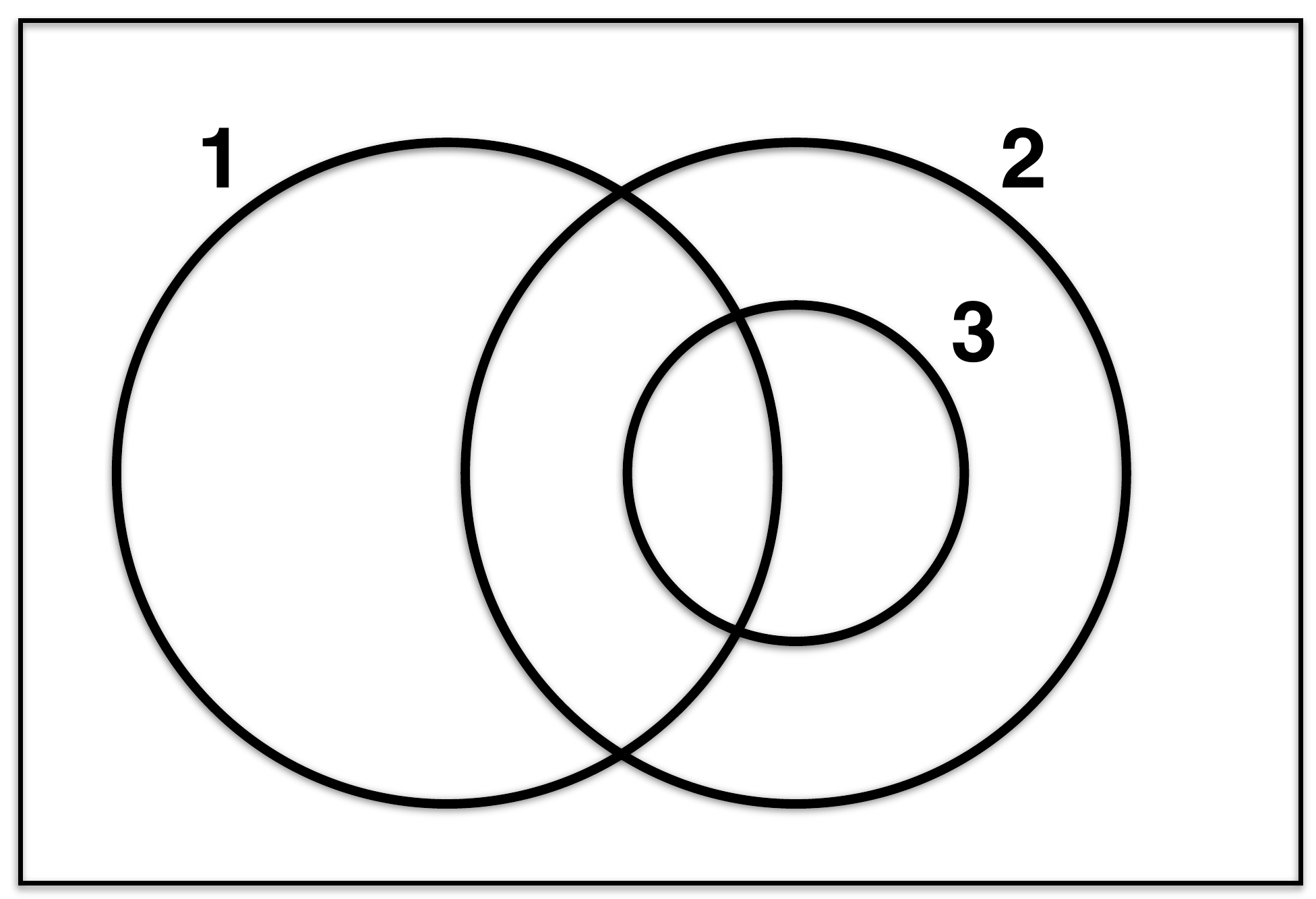} 
\caption{A diagram with a 1-piercing.}
\label{fig:1-piercingexample}
  \end{subfigure} 
  \caption{Examples of piercings.}
  \label{fig:examplesofpiercings} 
\end{figure}

\end{ex}


\begin{ex}[1-piercing]\label{ex:1-piercing} A 1-piercing intersects exactly one other curve in the diagram. Consider the diagram $d$ pictured in Figure \ref{fig:1-piercingexample}.  The abstract description $\D$ for this diagram has curve labels $\LL=\{1,2,3\}$ and zones $\Z=\big\{\emptyset,\{1\},\{2\},\{1,2\},\{2,3\},\{1,2,3\}\big\}$.  
Curve 3 is a 1-piercing of curve 1 identified by the zone $\{2\}\in \Z$ since
\begin{enumerate}
\item 3, 1 $\notin \{2\}$,
\item $\X_3=
\big\{\{2,3\},\{1,2,3\}\big\}
=\big\{(\{2\}\cup\{3\})\cup \Lambda_i\ |\ \Lambda_i\subseteq \{1\}\big\}=\Y_{\{2\}\cup\{3\},\{1\}}$, and
\item $\Y_{\{2\},\{1\}}=\big\{\{2\}\cup \Lambda_i\ |\ \Lambda_i\subseteq\{1\}\big\}=\big\{\{2\},\{1,2\}\big\}\subseteq \Z$.
\end{enumerate}
It should be noted that intersecting exactly one other curve is not sufficient to indicate a 1-piercing; in this example, curve 2 only intersects curve 1, but curve 2 is not a 1-piercing.
\end{ex}

In terms of drawings, we can think of a $k$-piercing as a curve that pierces $k$ other curves and splits $2^k$ zones.  These $2^k$ zones appear in the abstract description in the following way.

\begin{lemma}\label{lem:zonecount} Let $\D = (\LL,\Z)$ be an abstract description.  Let $\Lambda=\{ \lambda_1,...,\lambda_{k} \} \subseteq \LL$ be distinct curve labels.  If $\lambda_{k+1} \in \LL$ is a $k$-piercing of $\Lambda$ in $\D$ then there exist exactly $2^k$ elements of $\Z$ that contain $\lambda_{k+1}$, i.e. $|\X_{\lambda_{k+1}}| = 2^k$.
\end{lemma}

\begin{proof}  The statement follows from the second condition in the definition of a $k$-piercing.
\end{proof}

In order to define what it means for an abstract description to be $k$-inductively pierced, we need to discuss the removal of piercing curves in terms of the abstract description.

\begin{defn}[Removal of a curve] Given an abstract description $\D=(\LL, \Z)$ with $\lambda \in \LL$, then we define
$$\D - \lambda = ( \LL \setminus \{ \lambda \}, \Z - \lambda ), $$
where $\Z - \lambda = \{ Z \setminus \{ \lambda \} \ : \ Z \in \Z \}.$
\end{defn}
When $\C$ is a neural code, we can similarly discuss the removal of a neuron.  We define 
$$\mathcal C - \lambda = \{ (c_1, \ldots, c_{\lambda -1}, \hat c_{\lambda}, c_{\lambda +1}, \ldots, c_n) \ : (c_1, \ldots, c_n) \in \mathcal C \}.$$
If we  consider $\mathcal  C$ as a $n \times m$ matrix, where $m = | \mathcal C |$, then we can obtain $\mathcal C - \lambda$ by deleting the $\lambda$th row.

\begin{defn}\label{def:inductivelypiercedabstractdescrip}  An abstract description $\D=(\LL, \Z)$ is \emph{$k$-inductively pierced} if  $\D$ has a 0-, 1-, $\ldots,$ or $k$-piercing $\lambda$ such that $\D - \lambda$ is $k$-inductively pierced.  We will call a code $\mathcal C$ $k$-inductively pierced if its abstract description $\D_{\mathcal C}$ is $k$-inductively pierced.
\end{defn}

Since we will be focused on diagrams drawn entirely with circles, we will restrict our attention to $0$-, $1$-, and $2$-inductively pierced descriptions (a 3-piercing cannot occur in $\R^2$ if all curves must be circles).  In \cite{SZHR11}, the authors introduce a subclass of 2-inductively pierced descriptions, called \emph{inductively pierced}.  They show that if $\D$ is an inductively pierced abstract description, then there exists a drawing $d$ of $\D$, composed entirely of circles, which can be drawn in polynomial time. Note that 2-inductively pierced is a weaker requirement on an abstract description than inductively pierced as defined in \cite{SZHR11}. Inductively pierced, however, implies 2-inductively pierced. 




\medskip

We end this section by noting that well-formed diagrams with no intersecting curves have $0$-inductively pierced descriptions.

\begin{prop} \label{prop:crossing} Let $\D$ be well-formed. An abstract description $\D$ is inductively $0$-pierced if and only if all curves in any well-formed realization of $\D$ do not intersect.
\end{prop}

\begin{proof}
 Suppose $\D$ is inductively 0-pierced.  Assume for the sake of contradiction that there exists a well-formed realization $d$ of $\D$ such that there exist two curves, $\lambda_1$ and $\lambda_2$, that intersect.  Since $\D$ is $0$-inductively pierced, we can remove $0$-piercings until $\lambda_1$ or $\lambda_2$ is a $0$-piercing of the remaining curves, thus, without loss of generality, let us assume $\lambda_1$ is a 0-piercing of $D$.  Condition (2) of Definition \ref{def:k-piercing} implies that there exists a zone $Z$, with corresponding codeword $z$, such that the curve $\lambda_1$ is contained entirely in $\cap_{i \in Z} U_i$.  Thus, we can zoom in on this crossing as illustrated in Figure \ref{fig:crossingofcurves}. 


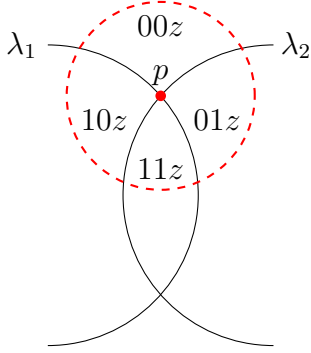
\begin{figure}[!ht]
\begin{center}
\begin{tikzpicture}
\draw[name path=1] (-1.5,0) arc (-90:90:2cm) node[anchor=east] {$\lambda_1$}; 
\draw[name path=2] (1.5,4) arc (90:270:2cm);
\node at (1.8,4) {$\lambda_2$};
\node at (-0.75,3) {$10{z}$};
\node at (0.75,3) {$01{z}$};
\node at (0,2.35) {$11{z}$};
\node at (0,4.25) {$00{z}$};
\path [name intersections={of = 1 and 2}];
	\coordinate[label=above:$p$] (p) at (intersection-2);
\fill[fill=red,inner sep=1pt,name intersections={of=1 and 2}]
    (intersection-2) circle (2pt) ;
\draw[red,thick,dashed] (p) circle (1.25cm);
\end{tikzpicture}
\end{center}
\caption{A closeup of a crossing of curves $\lambda_1$ and $\lambda_2$.}
\label{fig:crossingofcurves}
\end{figure}



From Figure \ref{fig:crossingofcurves}, we see that $|\X_{\lambda_1}|\geq 2$.  But since $\lambda_{1}$ is a 0-piercing, by Lemma \ref{lem:zonecount}, there exists exactly $2^0=1$ element of $\mathcal{Z}$ that contains $\lambda_{1}$, so $|\X_{\lambda_1}|=1$, a contradiction.  

For the converse, let $\D$ be a well-formed abstract description. Suppose for any well-formed realization $d$ of $\D$, no two curves intersect.   We will proceed by induction on $n$, the number of curves/curve labels.  The statement holds for $n=1$, since the only  code on one neuron is 0-inductively pierced.  Now suppose the statement holds for $n \leq r$, and let $n=r+1$. Since none of the curves $\lambda_1, \ldots, \lambda_n$ intersect, every pair of place fields $U_{i}$ and $U_{j}$ in $d$ are either disjoint or nested.  
Select a minimal field with respect to set inclusion, that is, select a place field $U_k$ such that for all $1 \leq i \leq n$ with $i \neq k$ either $U_i \cap U_k = \emptyset$ or $U_k\subset U_i$; note that since $n$ is finite and $\D$ is well-formed, such a place field must exist.  Then $\lambda_{k}$ is a 0-piercing of $\mathcal{D}$, and by the induction hypothesis, $\mathcal{D}- \lambda_k$ is $0$-inductively pierced.  Therefore, $\mathcal{D}$ is $0$-inductively pierced. 
\end{proof}

\subsection{The neural ideal and its canonical form}

As discussed in the Introduction, we will aim to identify $k$-inductively pierced codes using computational algebraic geometry. One way the neural code has been approached with an algebreo-geometric lens is through the neural ring and neural ideal \cite{neural_ring}.
To a given vector $v\in \{0,1\}^n$, we associate an indicator polynomial $\rho_v \in \ff_2[x_1,...,x_n]$ such that $\rho_v(v) =1$ and $\rho_v(v') = 0$ when $v'\neq v$.  The indicator polynomial is constructed as follows:
$$\rho_v = \prod_{v_i = 1} x_i \prod_{v_j=0}(1-x_j)$$  For example, if $n=3$ and $v= 101$, then $\rho_v = x_1x_3(1-x_2)$.  The indicator polynomial $\rho_v$ is a particular example of a {\it pseudo-monomial}, a polynomial of the form $\prod_{i\in \sigma} x_i \prod_{j\in \tau}(1-x_j)$ where $\sigma\cap\tau=\emptyset$.  Note that all monomials are necessarily pseudo-monomials by taking $\tau = \emptyset$.  For a given ideal $I$, we consider a pseudo-monomial $f$ to be {\it minimal} in $I$ if there are no pseudo-monomials $g\in I$ and $1\neq h\in \ff_2[x_1,...,x_n]$ such that $f=gh$; that is, $f$ is not a nontrivial multiple of another pseudo-monomial in $I$.

For any code $\C$ we define the {\it neural ideal} $J_\C$ as follows: $$J_\C = \langle \rho_v \ | \ v\in \{0,1\}^n\backslash \C\rangle$$  Note that for any $f\in J_\C$, we have $f(c) = 0$ for all $c\in \C$, i.e. all polynomials in $J_{\C}$ vanish on $\C$, but for any $v\notin \C$, there is at least one polynomial $g\in J_\C$ with $g(v)\neq 0$ and so $\C$ is precisely the variety of $J_\C$.  Since considering the full list of generators $\rho_v$ is often inefficient and opaque, we consider instead the {\it canonical form} of the neural ideal, $CF(J_\C)$, defined to be the set of minimal pseudo-monomials in $J_\C$.  For more about the neural ideal and the canonical form, see \cite{neural_ring}.

For our purposes, the most important property of the canonical form is its interpretation in terms of an arrangement of sets which realize the code in question.  We will make substantial use of the previously stated Proposition \ref{prop:RFstructure}(Lemma 4.2 from \cite{neural_ring}), which states that pseudo-monomials in $J_\C$ are in direct correspondence with set-theoretic information about the associated arrangement of sets $\mathcal U$ through the relation $$\prod_{i\in \sigma} x_i \prod_{j\in \tau}(1-x_j) \in J_\C\Leftrightarrow \bigcap_{i\in \sigma} U_i \subseteq \bigcup_{j\in \tau} U_j.$$ 
When $\tau=\emptyset$, this translates to  $\prod_{i\in \sigma} x_i \in J_\C$ if and only if $\bigcap_{i\in \sigma} U_i = \emptyset$. Since we always assume the all-zeros word is in the code, $J_\C$ will never contain a pseudo-monomial with $\sigma = \emptyset$. 

Importantly, these relationships hold regardless of the arrangement which is chosen.  That is, since the canonical form and the ideal $J_\C$ are properties of the code itself and not of the particular arrangement $\mathcal \U$ for which $\C = \C(\U)$, the presence of the pseudo-monomial $\prod_{i\in \sigma} x_i \prod_{j\in \tau}(1-x_j)$  in $J_\C$ indicates that  $\bigcap_{i\in \sigma} U_i \subseteq \bigcup_{j\in \tau} U_j$ in {\it any} arrangement $\U$ for which $\C = \C(\U)$.  

Beyond being minimal pseudo-monomials and often a condensed generating set for $J_\C$, the pseudo-monomials in $CF(J_\C)$ provide a minimal description of the interactions between the sets $U_i$.

\begin{ex} Let $\C =\{000,100,010,101\}$, the code from Example \ref{def:0-piercing}.  In this case, the canonical form is 
$$CF(J_\C) = \{ x_1x_2,\  x_2x_3, \ x_3(1-x_1) \}.$$ 
Using Proposition \ref{prop:RFstructure}, the elements give us the following place field relationships: $U_1\cap U_2 = \emptyset$, $U_2\cap U_3 = \emptyset$, and $U_3\subseteq U_1$.  While $U_1\cap U_2\cap U_3 = \emptyset$ is also true, corresponding to the fact that $x_1x_2x_3$ is also in $J_\C$,  the information $U_1\cap U_2=\emptyset$ implies the former fact and hence it is redundant information.
\end{ex}

 \subsection{The toric ideal of a neural code}\label{sec:background-toric-ideals}
The second algebraic object we will study is the toric ideal of $\mathcal C$. Let $\mathcal C=\{c_1, \ldots, c_m\}$ be a neural code on $n$ neurons and let $\C^*:= \calc \setminus \{ 00\ldots00 \}$, i.e. $\C^*$ is $\C$ with the all zeros word removed. Let $\mathbb{K}$ be a field and let $\phi_{\mathcal C}:\mathbb{K}[p_c\ |\ c\in\calc ^* \}]\longrightarrow\mathbb{K}[x_i\ |\ i\in\{1,\ldots,n\}]$ be the polynomial ring homomorphism defined by $$p_c\longmapsto\prod\limits_{i\in\supp(c)}^{}x_i.$$  
Recall that $I_{\calc}$, the toric ideal of $\calc$, is the kernel of the map $\phi_\calc$.  From standard results on toric ideals, the ideal $I_{\calc}$ is a prime ideal generated by binomials \cite{GB}.

\begin{ex}\label{ex:A1}
Let $\C =$ A1 = $\{100,010,001,110,101,011,111\}$, a neural code on 3 neurons.  We are using labeling of codes consistent with \cite{Curtoetal13}.  A place field diagram for this code is pictured in Figure \ref{fig:A1}. The toric ideal of this code $I_{\C} \subset \mathbb{K}[p_{100},p_{010},p_{001},p_{110},p_{101},p_{011},p_{111}]$ is generated by the following cubic and quadratics: $p_{111}-p_{100}p_{010}p_{001},\  p_{110}-p_{100}p_{010},\ p_{101}-p_{100}p_{001},$ and $p_{011}-p_{010}p_{001}$.

\end{ex}

\begin{ex}
 Let $\C =$ A2 = $\{100,010,110,101,111\}$, a neural code on 3 neurons.  A place field diagram for this code is pictured in Figure \ref{fig:A2}. The toric ideal of this code $I_{\C}$ is:
$$I_{\C}=\langle p_{111}-p_{010}p_{101}, p_{110}-p_{100}p_{010} \rangle.$$
\end{ex}

\begin{figure}[!ht] 
  \begin{subfigure}[b]{0.5\linewidth}
    \centering
    \includegraphics[width=0.7\linewidth]{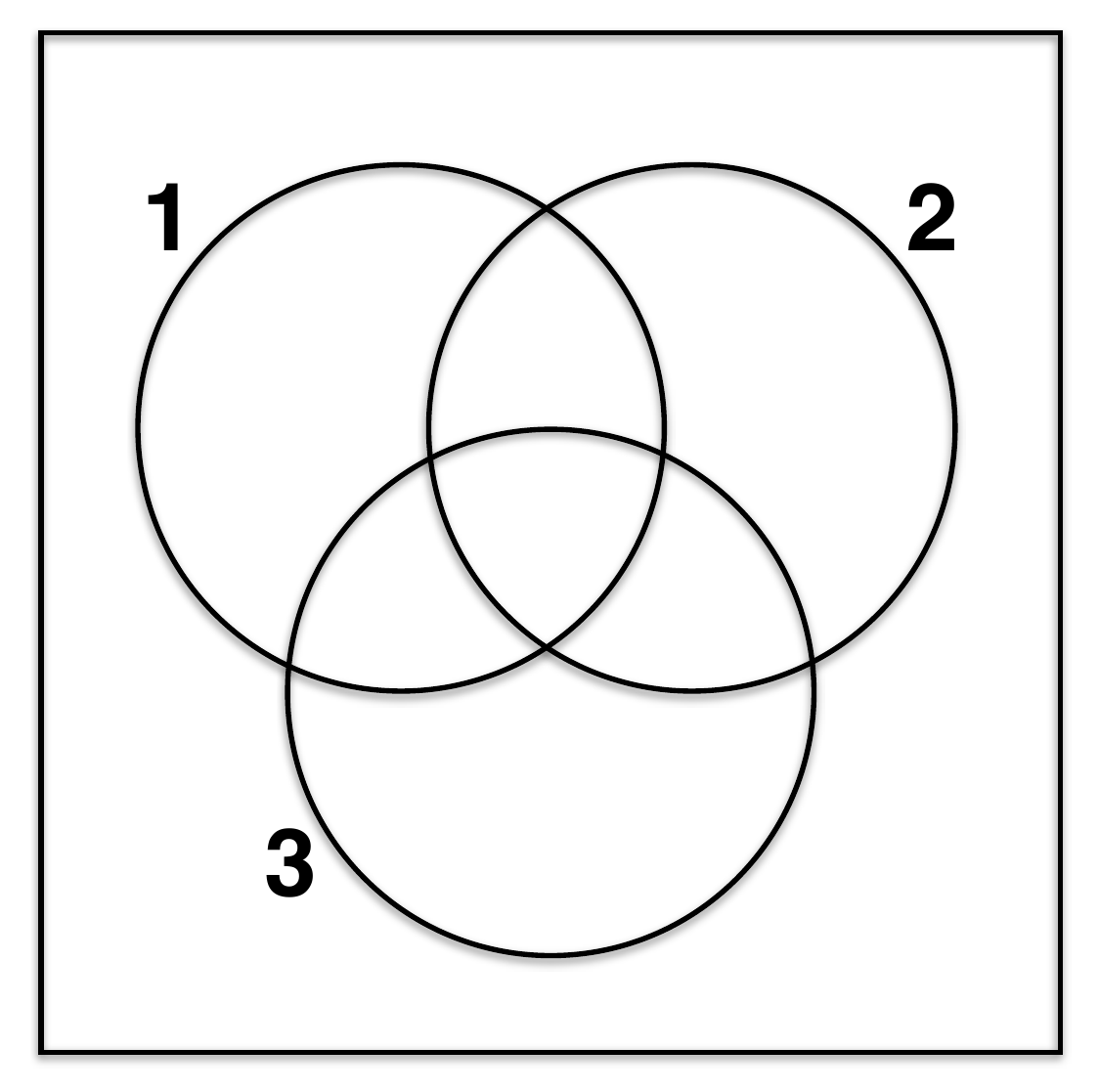} 
    \caption{A place field diagram of A1} 
    \label{fig:A1} 
  \end{subfigure}
  \begin{subfigure}[b]{0.5\linewidth}
    \centering
    \includegraphics[width=0.7\linewidth]{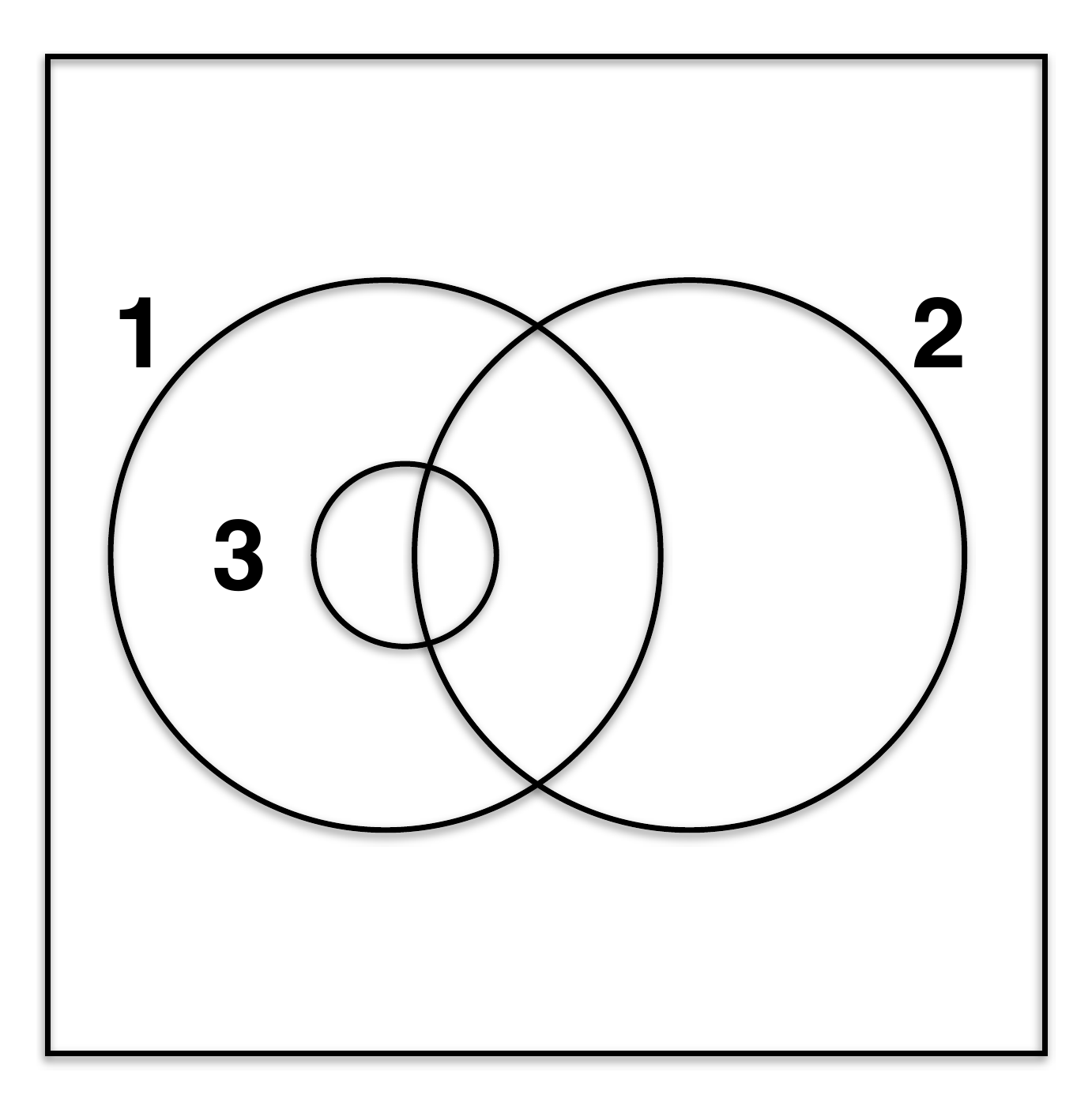} 
    \caption{A place field diagram of A2} 
    \label{fig:A2} 
  \end{subfigure} 
  \caption{Examples of place field diagrams of neural codes on three neurons}
  \label{fig:computingtoricideal} 
\end{figure}

Toric ideals have a nice combinatorial structure.  To visualize the information gathered from the toric ideal and to aid in proofs, we introduce the notion of a hypergraph, which is a generalization of a graph where edges can contain any number of vertices, not just two.  Toric ideals associated to hypergraphs have been studied in \cite{Vill,GP12,PS14}.

\begin{defn}\label{def:hypergraph} A \emph{hypergraph} $\mathcal{H}$ is a pair $\mathcal{H}=(V,E)$ where $V$ is a set of elements called nodes or vertices, and $E$ is a set of non-empty subsets of $V$ called hyperedges or edges.
\end{defn}

A code $\calc$ on $n$ neurons containing $m$ codewords can be visualized as a hypergraph with $n$ vertices and $m$ hyperedges.

\begin{defn}[Hypergraph of a neural code]
Given a code $\C \subset \{0,1\}^n$, the hypergraph associated to $\C$ is $\mathcal{H}_\C=(V, E)$ where $V=\{1,\ldots,n\}$ and $E=\{ Z_c\ |\ c \in \C^*\}$.  
\end{defn}

Every edge in $\mathcal H_{\C}$ corresponds to an indeterminate in $\mathbb{K}[p_c\ |\ c\in\calc ^* \}]$, and a collection, or multiset, of edges in $\mathcal H_{\C}$ corresponds to a monomial.  In order to encode a binomial of the form $p^u-p^v$, composed of a monomial $p^u$ with a positive sign and a monomial $p^v$ of a negative sign, we introduce edge colorings - specifically, bicolorings.




\begin{defn}[Bicoloring of a multiset of edges \cite{GP12}]
Let $\mathcal{E}$ be a multiset of edges from $\mathcal H$.  Partition $\mathcal{E}$ into two sub-multisets such that the multiset union of the two sub-multisets is equal to $\E$. Assign one color to each sub-multiset, say blue and red.  Then $\mathcal{E}=(\E_{\mathrm{blue}}, \E_{\mathrm{red}})=(B,R)$, where $B$ is the set of blue edges and $R$ is the set of red edges.  Such a coloring of $\mathcal{E}$ is called a \emph{bicoloring} of $\mathcal{E}$. 
\end{defn}

\begin{defn}[Balanced edge set {\cite{PS14}}]\label{def:balancededgesets} Let $\mathcal{E}=(B, R)$ be a multiset of bicolored edges from $\mathcal{H}=(V,E)$. For $v \in V$, let $\deg_B v$ be the number of edges in $B$ that contain $v$, counted with multiplicity.  Define $\deg_R v$ similarly. We say that $\mathcal{E}=(B,R)$ is \emph{balanced} with respect to the given bicoloring if for all $v \in V$ 
$$ \deg_B v = \deg_R v.$$
If $\mathcal{E}$ is balanced, we call $\mathcal{E}$ a \emph{balanced edge set} in $\mathcal{H}$.
\end{defn}


We say a binomial $f_\mathcal{E}$ arises from $\mathcal{E}$ if it can be written as
$$f_{\mathcal{E}} = \prod\limits_{e \in \mathcal{E}_{\mathrm{blue}}}^{}p_e -  \prod\limits_{e' \in \mathcal{E}_{\mathrm{red}}}^{}p_{e'}.$$ Every binomial in $I_\C$ arises from a balanced edge set in the hypergraph $\mathcal{H}_\C$ in this manner (see \cite{PS14} and \cite{GP13}).  


\begin{defn}[Primitive balanced set \cite{GP13}]
A balanced edge set $\mathcal{E}=(B,R)$ is \emph{primitive} with respect to $\mathcal H$ if there does not exists another balanced edge set $\mathcal{E'}=(B',R')$ in $H$ such that $B ' \subsetneq  B$ and $R' \subsetneq R$. \end{defn}

Primitive balanced edge sets in $\mathcal H _{\C}$ correspond to primitive binomials in $I_{\C}$.  A binomial $p^u-p^v \in I_{\C}$ is \emph{primitive} if there exists no other $p^{u'}-p^{v'} \in I_{\C}$ such that $p^{u'} | p^{u}$ and $p^{v'} | p^{v}$.  In particular, since $I_{\C}$ is prime, primitivity of a binomial $p^u-p^v \in I_{\C}$ implies the supports of $p^u$ and $p^v$ are disjoint. The set of all primitive binomials is a generating set of $I_{\C}$, thus, the set of all binomials arising from primitive balanced edge sets of $\mathcal H_{\C}$ generate $I_{\C}$.

\begin{ex}
As an example, consider the code $\C=B1$=$\{000,100,010,001,110,011\}$.  We can visualize the information from $\calc$ using the hypergraph illustrated in Figure \ref{ref:B1}.

 \begin{figure}[!ht]
 \begin{center}
 \includegraphics[width=2.5in]{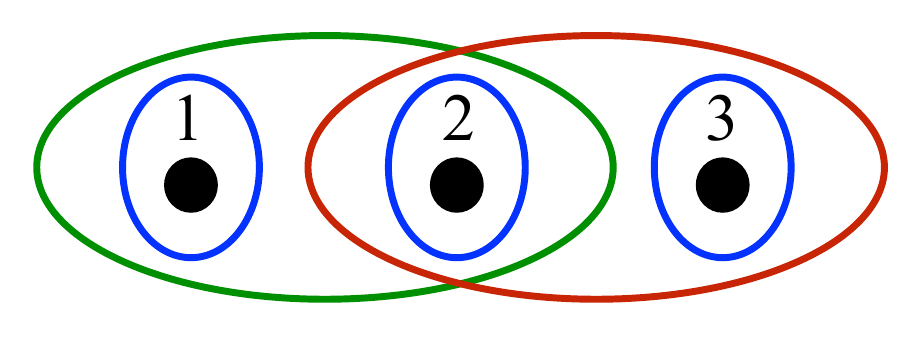}

 \caption{Hypergraph of neural code B1 
 }
 \label{ref:B1}
 \end{center}
 \end{figure}




By coloring the edges in the hypergraph, we see that there are at least two balanced edge sets: each of vertex 1 and vertex 2 are contained in a single blue edge and in a single green edge.  The blue edge around vertex 1 corresponds to the codeword 100, the blue edge around vertex 2 corresponds to the codeword 010, and the green edge around vertices 1 and 2 corresponds to the codeword 110.  Additionally, each of the vertices 2 and 3 are contained in a single blue edge and in a single blue edge.  In this case, the red edge corresponds to the codeword 011.  The generators of the toric ideal of B1 can be read off the diagram as $p_{110}-p_{100}p_{010}$ and $p_{011}-p_{010}p_{001}$.  

\end{ex}

Since hypergraphs combinatorially encode binomials in $I_{\C}$, we can obtain a degree bound on a minimal generating set of $I_{\C}$ by noting special properties of $\mathcal H_{\C}$.

\begin{defn}
Let $\mathcal E = (B, R)$ be a balanced edge set of $\mathcal H$. Let $\Gamma_1$ and $\Gamma_2$ be balanced edge sets of $\mathcal H$.  A multiset $S$ with support in $\mathcal H$ is a \emph{proper splitting set of $\E$ with decomposition $(\varGamma_1, S, \varGamma_2)$} if $S=\Gamma_{1_{red}} \cap \Gamma_{2_{blue}}$, $S \neq \Gonered, \Gtwoblue$, $\E \cup S = (\Gamma_1 \cup \Gamma_2) \setminus S$, and the following coloring conditions hold $\Goneblue, \Gtwoblue \subseteq B \cup S$ and $\Gonered, \Gtwored \subseteq R \cup S$.
\end{defn}

\begin{rmk}  When we are working with multisets, as in this section, we will interpret $\cup$, $\cap$, and $\subseteq$ in terms of their multiset definitions as described in Section 2 of \cite{GP12}.
\end{rmk}

We use the following degree bound theorem in Section 4.

\begin{thm}  \cite[Theorem 5.1]{GP13} \label{thm:combinatorialdegreebound}
\label{thm:nonuniform}  Let $\mathcal C$ be a code with corresponding hypergraph $\mathcal H_{\C}$.  The toric ideal $I_{\C}$ is generated by quadratics if for every primitive balanced edge set $\E$ of $H_{\C}$ with $|\E_{blue}| \geq |\E_{red}|$ and $|\E_{blue}|=n>2$, there exists a proper splitting set $S$ of $\E$ with decomposition $(\varGamma_1, S, \varGamma_2)$ where $|\varGamma_{i_{blue}}|, |\varGamma_{i_{red}}| < n$ for $i=1,2$.
\end{thm}

\medskip
\medskip

\section{The Neural Ideal and the Theory of Piercings}\label{sect:canonicalform}
Stapleton \emph{et al.} \cite{SZHR11} show that not all codes are inductively pierced and give some immediate ``red flags" that show a curve is not a piercing of any others. Under our algebraic interpretation, we find that the effect of a $k$-piercing on the canonical form of the neural ideal is very predictable, and thus it is possible  to detect algorithmically whether or not a particular code has any piercings.  Detecting whether a code is $k$-inductively pierced is more complicated, since it requires us to find an ordering on the elements.  However, if only $0$- and $1$-piercings are used, we are able to algebraically detect inductive piercings and determine a possible ordering using the canonical form of the neural ideal.

We first address the issue of simply detecting whether a $k$-piercing is present at all. We can directly translate the three conditions from Definition \ref{def:k-piercing} into set containment rules, and thus directly into canonical form language, obtaining the following result:

\begin{prop}\label{prop:translation} $\lambda_{k+1}$ is a piercing of $\lambda_1,...,\lambda_k$ identified by $Z$ if and only if the following hold:  

\begin{enumerate}

\item  $\lambda_i\notin Z$ for $1\leq i\leq k+1$
\item 
\begin{enumerate}
\item $x_{\lambda_{k+1}}x_i \in CF(J_\C)$ for all $i \in [n]\backslash (Z\cup \{\lambda_1,...,\lambda_{k+1}\})$
\item $x_{\lambda_{k+1}} \prod_{i\in \sigma} x_i \notin CF(J_\C)$ for all $\sigma\subseteq \{\lambda_1,...,\lambda_k\}$ 
\item $x_{\lambda_{k+1}} (1-x_j) \in CF(J_\C)$ for all $j\in Z$
\end{enumerate}
\item If we reduce the canonical form by setting $x_i = 0$ for $i\in [n]\backslash (Z \cup \{\lambda_1,...,\lambda_{k+1}\})$ and setting $x_j=1$ for all $j\in Z$, we obtain only zeros.
\end{enumerate}
\end{prop}

The translation of the definition of $k$-piercing into the language of the canonical form follows from prior work \cite{neural_ring}; note that the numbered conditions here match the numbered conditions from Definition \ref{def:k-piercing}.  When $k$, $\{\lambda_i\}, \lambda_{k+1}$, and $Z$ are known, it is not a difficult matter to check the above conditions.  However, even if none of this information is known {\it a priori}, we can still determine if the code has piercings by checking each $\lambda$ in turn.

To check if a particular $\lambda$ is a piercing, we consider the polynomials in $CF(J_\C)$ which involve $\lambda$.  Let $A$ be the set of indices $a$ such that $x_{\lambda}x_a$ is in $CF(J_\C)$, and let $Z$ be the set of indices $j$ such that $x_\lambda(1-x_j)$ is in $CF(J_\C)$; as $CF(J_\C)$ is a canonical form, $A$ and $Z$ are disjoint.  Then set $B = [n] \backslash (A\cup Z)$.  Set $k= |B|$, arbitrarily order $B$ as $B = \{\lambda_1,...,\lambda_k\}$, and check the conditions above. If they hold, then $\lambda$ is a $k$-piercing of $B$ within the zone $Z$. If they do not, then $\lambda$ is not a piercing of any curves in any zone.  Note that by the conditions necessary in Proposition \ref{prop:translation}, particularly condition (2), the possible zone $Z$ and the curves $B$ which are pierced by $\lambda$  are uniquely determined. \\


We now consider the special case of $0$- and $1$- piercings.  As we will show, a 0-inductively pierced code has a canonical form of a very specific kind: the relationship between two place fields $U_i$ and $U_j$ is either disjointness $(x_ix_j \in CF(J_\C))$ or containment in a unique direction (either $x_i(1-x_j)$ or $x_j(1-x_i)$ are in $CF(J_\C)$, but not both, as this would imply equality of $U_i$ and $U_j$ which is impossible in a well-formed code). 

Following the language in \cite{SZHR11} and our previous discussion, recall that a code is \emph{0-inductively pierced} if $\C$ has a $0$-piercing $\lambda$ such that $\C-\lambda$ is 0-inductively pierced.  In particular, the only 0-inductively pierced code on $1$ neuron is $\{0,1\}$. If $\C$ has a $0$-piercing $\lambda$, then this means that there is a codeword $c\in \C$ and a background zone $Z$ such that if $c_\lambda = 1$ then $\supp(c) = Z\cup\{\lambda\}$ (intuitively, $U_\lambda$ is properly contained within zone $Z$ in any diagram).  


\begin{lemma}\label{lem:onecodeword} If $\lambda$ is a 0-piercing for a code $\C$, then $\C$ can be obtained from $\C-\lambda$ by adding a neuron which is always 0, and then adding a codeword $v$ such that $\supp(v) = \supp(c)\cup\{\lambda\}$ for some $c\in \C$.
\end{lemma}

\begin{proof} 
Suppose $\lambda$ is a $0$-piercing of $\C$, and $Z_c$ is a zone such that the definition holds.  This implies that we have the three conditions listed in Example \ref{def:0-piercing}.
Thus, a code $\C$ with a $0$-piercing $\lambda$ is a code where all codewords except $c$ have $c_\lambda = 0$, and there is exactly one codeword whose support is identical to $c$ except at $\lambda$.
\end{proof}

Similarly, translating the definitions of 1- and 2-piercing, we have the following results:

\begin{lemma} If $\lambda$ is a 1-piercing of $\{\lambda_1\}$ for a code $\C$, then $\C$ can be obtained from $\C-\lambda$ by adding a neuron which is always 0, and then adding two codewords: one whose support is $\supp(c)\cup\{\lambda\}$ for some $c\in \C$, and one whose support is $\supp(c)\cup\{\lambda, \lambda_1\}$.
\end{lemma}

\begin{lemma} If $\lambda$ is a 2-piercing of $\{\lambda_1, \lambda_2\}$ for a code $\C$, then $\C$ can be obtained from $\C-\lambda$ by adding a neuron which is always 0, and then adding four codewords: one whose support is $\supp(c)\cup\{\lambda\}$ for some $c\in \C$, and then codewords with support $\supp(c)\cup\{\lambda, \lambda_1\}$, $\supp(c)\cup \{\lambda, \lambda_2\}$, and $\supp(c)\cup\{\lambda, \lambda_1,\lambda_2\}$.
\end{lemma}

For each of these cases, we can translate these operations into changes in the canonical form, as shown in the following result. 

\begin{thm} \label{thm:cfchange}  Let $\C$ be a code on $n$ neurons.  Then, the following three statements hold:

\begin{enumerate}
\item If $\lambda$ is a 0-piercing at zone $z$, then $$CF(J_\C) = CF(J_{\C-\lambda}) \cup\{x_ix_\lambda\ | \ i \in [n]\backslash (z\cup\{\lambda\}) \} \cup\{x_\lambda(1-x_j)\ | \ j\in z\}$$

\item If $\lambda$ is a 1-piercing of $\{\lambda_1\}$ at zone $z$, then $$CF(J_\C) = CF(J_{\C-\lambda}) \cup\{x_ix_\lambda\ | \ i \in [n]\backslash (z\cup\{\lambda, \lambda_1\} )\} \cup\{x_\lambda(1-x_j)\ | \ j\in z\}$$

\item If $\lambda$ is a 2-piercing of $\{\lambda_1, \lambda_2 \}$ at zone $z$, then $$CF(J_\C) = CF(J_{\C-\lambda}) \cup\{x_ix_\lambda\ | \ i \in [n]\backslash (z\cup\{\lambda, \lambda_1, \lambda_2\} )\} \cup\{x_\lambda(1-x_j)\ | \ j\in z\}$$

\end{enumerate}
\end{thm}

The particular details of how changes in the canonical form reflect changes to the code are discussed at length in \cite{CY15}; in our case, the change is the removal of a neuron whose place field pierces $k$ others.  This leads to the following property  of the canonical form of a $k$-inductively pierced code and a characterization of 0-inductively pierced codes.

\begin{thm} \label{thm:canonicalformat} Let $\C$ be a code on $n$ neurons. 
\begin{enumerate}
\item If $\C$ is $k$-inductively pierced, then the $CF(J_\C) \subseteq \{f_{ij} \,|\, 1\leq i<j\leq n, f_{ij}\in \{x_ix_j, x_i(1-x_j), x_j(1-x_i)\}\}$.
\item $\C$ is $0$-inductively pierced if and only if $CF(J_\C) = \{f_{ij} \,|\, i,j\in [n], i<j, f_{ij}  \in \{x_ix_j, x_i(1-x_j), x_j(1-x_i) \}$.   
\end{enumerate}
\end{thm}

\begin{proof} 

Both (1) and the forward direction of (2) follow immediately by induction from Theorem \ref{thm:cfchange}, and the fact that the only k-inductively pierced code on $1$ neuron is $\C = \{0,1\}$, which has empty canonical form.

To see the reverse direction of (2), we proceed by induction on $n$.  If $n=1$, then $CF(J_\C) = \emptyset$ so $\C = \{0,1\}$, which is 0-inductively pierced. Now, let $n>1$, and assume the result holds for $n-1$.  We will show that there is some $\lambda$ so that $CF(J_{\C-\lambda})$ is also in the desired form, and that $\lambda$ is a $0$-piercing of $\C$. To do this, pick $\lambda$ which maximizes $|B_\lambda|$, where $B_\lambda = \{i \,|\, x_\lambda(1-x_i) \in CF(J_\C) \}$ (this choice may not be unique).  Let $B=B_\lambda$, and let $v_B$ be the vector with $\supp(v_B) = B$.  To see that $\lambda$ is a $0$-piercing of $\C$, we will show that both $v_B$ and $v_{B\cup \lambda} \in \C$, and that if $c\in \C$ with $c_\lambda = 1$, then $\supp(c)  =\{ \lambda \} \cup B$. 

First, to show that $v_B$ and $v_{B\cup \lambda}$ are in $\C$, we will show that for each $f_{ij} \in CF(J_\C)$, we have $f_{ij}(v_B) = f_{ij}(v_{B\cup \lambda}) = 0$.  We have two cases: 

(\emph{Case 1: $f_{ij} = x_ix_j$}).  If either $i$ or $j$ is not in $B\cup\lambda$, then $f_{ij}(v_B) = f_{ij}(v_{B\cup\lambda}) = 0$. We will show that any options where both $i$ and $j$ are in $B\cup \lambda$ is impossible.  
 If both $i$ and $j$ are in $B$, then  $x_\lambda(1-x_i)$ and $x_\lambda(1-x_j)$ are both in $CF(J_\C)$. But then $x_\lambda = x_\lambda(1-x_i) + x_i(x_\lambda(1-x_j)) + x_\lambda(x_ix_j) \in J_\C$ which is a contradiction.  If $i\in B$ and $j=\lambda$, then $x_\lambda(1-x_i)\in CF(J_\C)$, which is a contradiction since we can't have both $x_\lambda(1-x_i)$ and $f_{i\lambda} = x_\lambda x_i$ in $CF(J_\C)$. A similar argument holds for $i= \lambda$ and $j\in B$.  Thus, any case where $i$ and $j$ are both in $B\cup \{ \lambda\}$ is impossible. 

(\emph{Case 2: $f_{ij} = x_i(1-x_j)$}).  If $i\notin B\cup \{ \lambda \}$, then clearly $f_{ij}(v_B) = f_{ij}(v_{B\cup \lambda}) = 0$.  If $i=\lambda$, then $j$ is necessarily in $B$ by definition, and hence $f_{ij}(v_B) = f_{ij}(v_{B\cup \lambda}) = 0$.  Finally, if $i\in B$, then $j\neq \lambda$ as we cannot have both $x_\lambda(1-x_i)$ and $x_i(1-x_\lambda)$ in $CF(J_\C)$ by the presumed format.  Thus, in this case, $(1-x_j)(x_\lambda(1-x_i))+ x_\lambda(x_i(1-x_j)) = x_\lambda(1-x_j)\in J_\C$, and as neither $x_\lambda$ nor $1-x_j$ may be in $J_\C$, we have $x_\lambda(1-x_j)\in CF(J_\C)$, and hence $j\in B$.  So again $f_{ij}(v_B) = f_{ij}(v_{B\cup\lambda}) =0$.\\

Now, suppose by way of contradiction that there is some $c_\lambda = 1$ but $\supp(c) \neq \{B\cup \{\lambda\}\}$. If $i\in B$ but $c_i\neq 1$, then $x_\lambda(1-x_i)$ could not be in $J_\C$ as $c$ would not evaluate to $0$ on it, but we know $x_\lambda(1-x_i)$ is in the canonical form by definition of $B$.  So this cannot occur.  Thus, $\supp(c)\neq \{B\cup\{\lambda\}\}$ must imply $\supp(c)\supsetneq B\cup\{ \lambda\}$.  Let $i\in \supp(c)$ and $i\notin B\cup\{\lambda\}$.  For every $j\in B$, we know $x_ix_j\notin CF(J_\C)$, and if $x_j(1-x_i)\in CF$ then $i\in B$ so we must have $x_i(1-x_j)\in CF(J_\C)$.  Furthermore, $x_i(1-x_\lambda) \in CF(J_\C)$ also, as $x_\lambda(1-x_i)$ and $x_ix_\lambda$ cannot be.  So $\lambda$ was not a maximum choice which is a contradiction.

This shows that $\lambda$ is indeed a 0-piercing of $\C$; i.e.,  $c_\lambda = 0$ for all codewords except one, and that one is a copy of $v$ otherwise.  When we delete neuron $\lambda$, we obtain $CF(J_{\C-\lambda})$, which by \cite{CY15} contains exactly the elements of $CF(J_\C)$ which did not involve $x_\lambda$; hence each pair will occur exactly once and thus $CF(J_{\C-\lambda})$ will have the desired format.
\end{proof}

From these results, we see that it is possible to detect whether a code is $0$-inductively pierced by using the canonical form. We now move on to detecting whether a code is $1$-inductively pierced. Recall that a code $\C$ is \emph{1-inductively pierced} if either $\C = \{0,1\}$, or there is some 0- or 1-piercing $\lambda$ such that $\C - \lambda$ is $0$ or $1$-inductively pierced.


For a given code $\C$, we define the \emph{general relationship graph} $G(\C)$ to have vertices $V=[n]$, and an edge $\{i,j\}$ appears if and only if none of $x_ix_j, x_i(1-x_j)$ and $x_j(1-x_i)$ appear in $CF(J_\C)$. That is, $G(\C)$ connects two vertices exactly when there is no interesting relationship (disjointness, containment) between their respective place fields.

\begin{ex} Let $\C = \{00000, 10000, 11000, 10100, 11100, 01000, 00010, 01010, 01011\}$. The canonical form for this code is $CF(J_\C) = \{x_1x_4, x_1x_5, x_3x_4, x_3x_5, x_3(1-x_1), x_5(1-x_2), x_5(1-x_4)\}$.  The general relationship graph $G(\C)$ places an edge between those pairs that do not appear together in an element in the canonical form; here, those edges are $\{1,2\}, \{2,3\}, \{2,4\}$. See Figure \ref{figure:C2} for an Euler diagram for this code and a drawing of $G(\C)$.

\begin{figure}[h]
\begin{center}
\includegraphics[width=.7\textwidth]{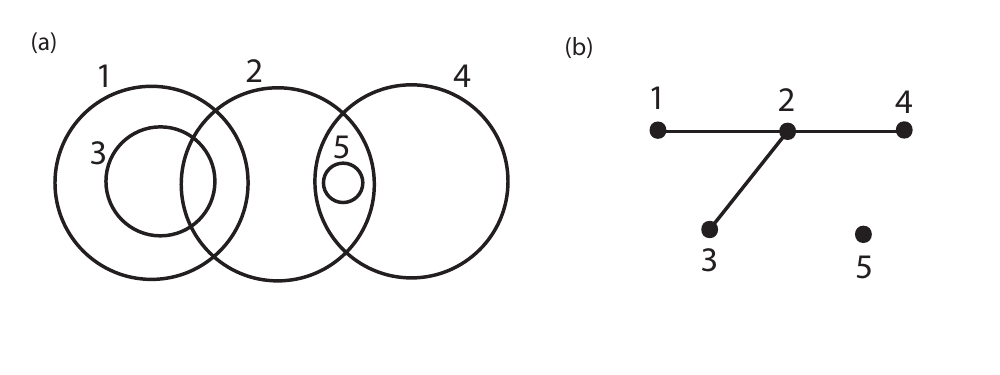}
\end{center}
 \caption{{\bf (a)} Euler diagram for the code $\C$. {\bf(b)}The general relationship graph $G(\C)$.}
 \label{figure:C2}
\end{figure}

\end{ex}

\begin{thm}\label{thm:01piercing} $\C$ is 1-inductively pierced if and only if $CF(J_\C)$ consists only of degree two pseudo-monomials meeting the following conditions:
\begin{enumerate}[(i)]
\item For each pair $\{i,j\}$, at most one of $x_ix_j, x_i(1-x_j), x_j(1-x_i)$ appears in $CF(J_\C)$.
\item $G(\C)$ is a forest.
\end{enumerate}

\end{thm}

\begin{cor} $\C$ is 0-inductively pierced if and only if $\C$ is 1-inductively pierced and $G(\C)$ consists only of isolated vertices.
\end{cor}

This follows by Theorem \ref{thm:canonicalformat}.

\begin{proof}  First, we show the conditions are necessary. By Theorem \ref{thm:canonicalformat}, any $1$-inductively pierced code has a canonical form with only degree-two polynomials using each pair $i,j$ at most once, so condition (i) follows.  Condition (ii) follows by induction: the only 1-inductively pierced code on 1 neuron is $\{0,1\}$ and the graph of this code is a single vertex.  If $\lambda$ is a 0- or 1-piercing of $\C$ and $\C-\lambda$ is 1-inductively pierced, then $G(\C-\lambda)$ is a forest by induction.  If $\lambda$ is a $0$-piercing, then by Theorem \ref{thm:cfchange}, $CF(J_\C)$ will contain a pair $x_\lambda x_j$ or $x_\lambda(1-x_j)$ for every element $j\in [n]-\lambda$, plus the relationships in $CF(J_{\C-\lambda})$, so $G(\C)$ is obtained from $G(\C-\lambda)$, by adding an isolated vertex $\lambda$.  If $\lambda$ is a 1-piercing of $\lambda_i$, then again by Theorem \ref{thm:cfchange}, $CF(J_\C)$ will have a pair $x_\lambda x_j$ or $x_\lambda(1-x_j)$ for every element $j\in [n]-\{\lambda,\lambda_i\}$. so $G(\C)$ is obtained from $G(\C-\lambda)$ by adding a vertex $\lambda$, and adding single edge connecting $\lambda$ with $\lambda_i$, so $G(\C)$ is also a forest.\\

To show these two conditions are sufficient, we proceed by induction on $n$. For $n=1$, then the only code which has no degree 1 terms in its canonical form is the code $\{0,1\}$, whose canonical form is empty. In this case, condition (i) is trivially satisfied, and $G(\{0,1\})$ an isolated vertex, so the code meets meets condition (ii).  This code is trivially 1-inductively pierced.

Now, let $\C$ be an arbitrary code on $n$ vertices meeting the required conditions.  We will show that there is some vertex $\lambda$ of $G(\C)$ with degree $\leq 1$ such that $x_j(1-x_\lambda)$ does not appear in $CF(J_\C)$ for any $j \in [n]\backslash\{\lambda\}$.  Once we prove that such a $\lambda$ must exist, we will show that it is either a $0$ or $1$ piercing of $\C$; furthermore, that $\C-\lambda$ meets the required conditions and hence by induction $\C-\lambda$ is 1-inductively pierced.

First, we prove that such a $\lambda$ must exist.  Let $L$ be the set of elements of $[n]$ which have degree $\leq 1$ in $G(\C)$; since $G$ is a forest and $n\geq 2$, $L$ contains at least two elements.  Suppose by way of contradiction that for every element $\ell \in L$ there is some $i\in [n]$ such that $x_i(1-x_\ell)\in CF(J_\C)$. We will show there is some $\ell'\in L$ with $x_\ell'(1-x_\ell)\in CF(J_\C)$.   To do so, we will follow a path in $G(\C)$ starting from $i$ of vertices $j$ with $x_j(1-x_\ell)\in CF(J_\C)$, until we reach a leaf.  If $i\in L$, then we're done. If not, then $i$ has degree $\geq 2$, and so if $i$ and $\ell$ are connected in $G(\C)$, the degree of $i$ allows us to move along a path away from $G(\C)$ towards any leaf; if they are not connected, we may move along any path towards any leaf.  Let $\{i,j\}$ be the first edge along this path. Note that $\{j,\ell\}$ is not an edge, by our choice of path. Hence, one of $x_\ell x_j, x_\ell(1-x_j)$, or $x_j(1-x_\ell)$ is in $CF(J_\C)$.   If $x_\ell x_j\in CF(J_\C)$, then we know $x_i(x_\ell x_j) + x_j( x_i(1-x_\ell)) = x_ix_j\in CF(J_\C)$, but this is a contradiction as $\{i,j\}$ is an edge.  If $x_\ell(1-x_j) \in CF(J_\C)$, then $(1-x_j)[x_i(1-x_\ell)] + x_i[x_\ell(1-x_j)] = x_i(1-x_j)\in CF(J_\C)$, again a contradiction as $\{i,j\}$ is an edge.  So it must be the case that $x_j(1-x_\ell)\in CF(J_\C)$.
 
  Repeating this as many times as necessary as we follow the path, we will eventually arrive at a leaf $\ell'$ with $x_\ell'(1-x_\ell)$, and thus find another element $\ell'$ of $L$  with $x_\ell'(1-x_\ell)$.  Since $L$ is finite, repeating this leaf-finding process gives a list $\ell_1,\ell_2,...,\ell_k$ of vertices with each of $\ell_i(1-\ell_{i+1})$ and $\ell_k(1-x_{\ell_1})$ appearing in $CF(J_\C)$.  Since for any triple $i,j,k$,  where $x_i(1-x_j)$ and $x_j(1-x_k)\in CF(J_\C)$, we have  $x_i(1-x_k)\in CF(J_\C)$, we then obtain that $\ell_1(1-x_k)\in CF(J_\C)$ as well as $x_{\ell_k}(1-x_1)$, and this contradicts condition (i). Thus, such a $\lambda$ must exist.

Now, given such a $\lambda$, since has degree $\leq 1$ in $G(\C)$, there is at most one $\lambda_i$ which there is no degree-2 polynomial involving $\lambda$ and $\lambda_i$. If there is none, then by condition (i), $\lambda$ is a $0$-piercing.  If there is exactly one, $\lambda$ is a 1-piercing of $\lambda_i$.  

Finally note that by Theorem \ref{thm:cfchange}, the canonical form $CF(J_{\C-\lambda})$ is exactly $CF(J_\C)$ with any term involving $\lambda$ removed; such a canonical form will still have only degree-2 pseudo-monomials and condition (i) will still be met.  Furthermore, $G(\C-\lambda)$ is just $G(\C)$ with $\lambda$ deleted since no other edges are affected, so $G(\C-\lambda)$ is a forest, so condition (ii) is still met.  Thus, by induction, $\C-\lambda$ is 0 or 1-inductively pierced and the result holds.
\end{proof}

The proof of the reverse implication above is quite powerful. Not only does it show that 1-inductively pierced codes can be detected, it actually gives us a way to determine a sequence $(\lambda_1, \ldots, \lambda_n)$ of curve labels such that $\lambda_{i}$ is a $0$ or $1$ piercing of $\mathcal D_{\calc - \lambda_{i+1} - \cdots - \lambda_n}$, which we will call a \emph{drawing order}. Formalizing the algorithm, we determine a possible drawing order as follows:

\begin{alg}\label{alg:ordering}
$\ $\\

{\bf Input:}  $CF(J_\C)$ and $G(\C)$ meeting conditions (i) and (ii).

{\bf Output: }an 1-inductively pierced drawing order for $\C$.\\

Step 0: Initialize an empty list $L$.

Step 1: Find a vertex $\ell$ in $G(\C)$ such that no pseudo-monomial of the form $x_j(1-x_\ell)$ is present in $CF(J_\C)$, and $\deg(\ell)\leq 1$.

Step 2: Set $L = (\ell, L)$.

Step 3: If $|L| = n$, stop.  Otherwise, set $\C  = \C-\ell$, recompute $CF(J_\C)$ and $G(\C)$, and return to Step 1.\\

\end{alg}

Note that this algorithm will serve to give an acceptable ordering even if the code is 0-inductively pierced.  Section 5 contains an example of a drawing order outputted from the above algorithm.\\

We believe similar results can be continued to $k$-inductively pierced codes, as in the following conjecture:
\begin{conj*} A code is $k$-inductively pierced if and only if $CF(J_\C)$ consists only of degree two pseudo-monomials meeting the following conditions:
\begin{enumerate}[(i)]
\item For each pair $\{i,j\}$, at most one of $x_ix_j, x_i(1-x_j), x_j(1-x_i)$ appears in $CF(J_\C)$.
\item $G(\C)$ is a chordal graph with no $k+2$-cliques.
\end{enumerate}
\end{conj*}

The conditions in this conjecture are certainly necessary.  A $k$-inductively pierced code has an ordering of $0,1,2$,... $k$ piercings; a quick investigation finds that when removing a $k$-piercing, the graph $G(\C)$ changes in a very predictable way. In particular, we remove a vertex $v$ of degree $k$ whose neighborhood, along with $v$, forms a $k+1$-clique. The inductive ordering of the piercings thus exhibits a {\it perfect elimination ordering} \cite{West} for the graph $G(\C)$ which implies that $G(\C)$ must be chordal, and by the  degree of each removed vertex we know the graph will be $k+1$-colorable and hence contain no $k+2$-clique.  We believe that a similar proof technique to that used in Theorem \ref{thm:01piercing} might be able to extend more generally to prove this conjecture; however, since not even all 2-inductively pierced codes are realizable, it might be more productive to characterize the graphs of inductively pierced codes rather than the more general classes of $k$-inductively pierced codes.  Likewise, while the results of Stapleton, \emph{et al.}, \cite{SZHR11} imply that any code which is 1-inductively pierced can be realized with circular place fields, they also show that 2-inductively pierced codes require additional conditions for realizability that are not satisfied by every 2-inductively pierced code. It is not yet known if these additional conditions can be translated into clear conditions on the canonical form.

\section{Toric Ideals of $k$-inductively Pierced Neural Codes}\label{sect:toricideal}

In this section, we discuss necessary and sufficient conditions for $0$ and $1$-inductively pierced codes in terms of their corresponding toric ideals.  This gives a second computational algebraic geometry approach to checking whether codes are $k$-inductively pierced. We begin by investigating $0$-inductively pierced codes.

\begin{thm} \label{thm:0idealiff0pierced}
Let $\mathcal{C}$ be a code on $n$ neurons such that each neuron fires at least once, i.e. $\bigcup_{z \in \mathcal C} \ \supp(z)=[n]$.  Let $\mathcal C$ be well-formed.  Then, the toric ideal $I_{\mathcal C} = \langle 0 \rangle$ if and only if $\mathcal{C}$ is $0$-inductively pierced.
\end{thm}

In order to prove Theorem \ref{thm:0idealiff0pierced}, we will need  the following lemma that notes that crossings of curves in well-formed diagrams result in non-zero toric ideals. The proof is constructive, describing a quadratic binomial that appears in the toric ideal of a given neural code if a realization contains a crossing of curves.

\begin{lemma}\label{lem:generator} Let $\mathcal C \subset \{0,1\}^n$ be a neural code with abstract diagram $\D_{\mathcal C}$. If a well-formed diagram $d$ of $\D_{\mathcal C}$ contains two curves that intersect, then the toric ideal $I_{\mathcal C}$ is nonempty. \label{thm:crossinggenerator}
\end{lemma}

\begin{proof}
Let $d$ be a well-formed diagram of $\D_{\mathcal C}$ such that two curves $\lambda_1$ and $\lambda_2$ intersect.  Let $q$ be an intersection point of $\lambda_1$ and $\lambda_2$. Since $d$ is well-formed, there exists an open ball around $q$, that is contained entirely in a single zone $Z$ of $d -\lambda_1 - \lambda_2$ with associated codeword $z$ as illustrated in Figure \ref{fig:crossingofcurves}.  Thus, the following codewords must be in $\mathcal C$: $10z, 01z, 00z, 11z$. 
Then the matrix of the codewords in $\mathcal C$ looks like 
\[ \left(
\begin{array}{ccccc}
0 & 1 & 0 & 1 &\ldots \\
0 & 0 & 1 & 1 &\ldots \\
z^T & z^T & z^T & z^T & \ddots
\end{array}
\right), \] 
and we have 
\begin{align*}
& \phi_{\mathcal C}(p_{00z}) = \prod\limits_{i\in supp(z)}^{}x_i
& \phi_{\mathcal C}(p_{10z} ) =  x_1\prod\limits_{i\in supp(z)}^{}x_i \\
&  \phi_{\mathcal C}(p_{01z}) =  x_2\prod\limits_{i\in supp(z)}^{}x_i  &
 \phi_{\mathcal C}(p_{11z}) =  x_1x_2\prod\limits_{i\in supp(z)}^{}x_i.
\end{align*}
Therefore $p_{11z}p_{00z}-p_{10z}p_{01z}$ is an element of the toric ideal $I_{\mathcal C}$.  Indeed, $$\phi_{\mathcal C} (p_{11z}p_{00z}-p_{10z}p_{01z}) =\left(\prod\limits_{i\in supp(z)}^{}x_i\cdot x_1x_2\prod\limits_{i\in supp(z)}^{}x_i\right)-\left(x_1\prod\limits_{i\in supp(z)}^{}x_i\cdot x_2\prod\limits_{i\in supp(z)}^{}x_i\right)=0.$$  
\end{proof}


While our goal is to understand the realization of a code by understanding its toric ideal, we note that Lemma \ref{lem:generator} and its proof allow one to understand some things quickly about a toric ideal of a code simply by noticing motifs, i.e. place field configurations in its realization. Indeed, we can conclude that there is a quadratic binomial for every two fields that intersect transversally as in Figure \ref{fig:twofields}.  For example, the toric ideal of the code associated to a chain of $n$ fields  as illustrated in Figure \ref{fig:nfields} contains a quadratic binomials for each of its $n-1$ pairwise intersections.

\begin{figure}[!ht]
\includegraphics[width=.3\linewidth]{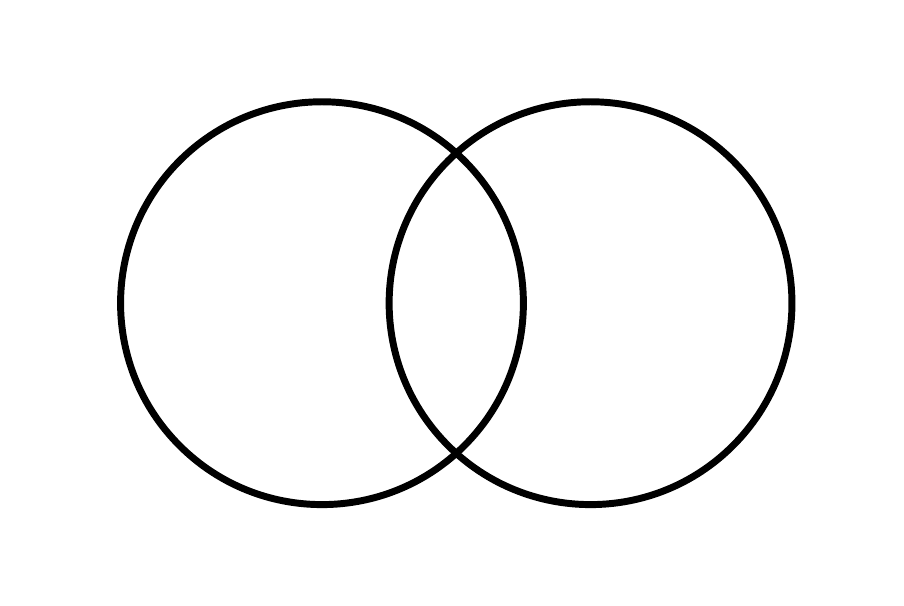}
\caption{Two fields that intersect transversally lead to a quadratic binomial in the toric ideal.}
\label{fig:twofields}
\end{figure}

\begin{figure}
\includegraphics[width=.7\linewidth]{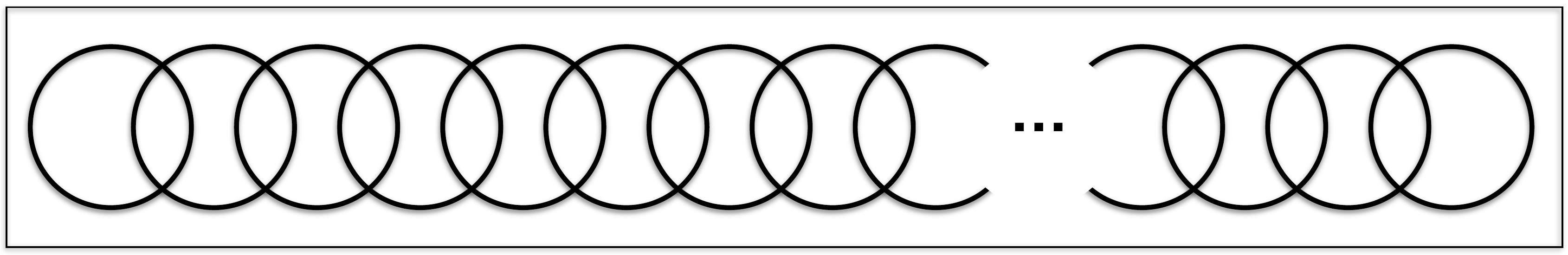}
\caption{A chain of $n$ fields intersecting transversally.}
\label{fig:nfields}
\end{figure}
Now we proceed to the proof of Theorem \ref{thm:0idealiff0pierced}, which states that the toric ideal of a code is trivial exactly when we fail to see the motif of Figure \ref{fig:twofields}.

\begin{proof}[Proof of Theorem  \ref{thm:0idealiff0pierced}]
Let $\mathcal{C}$ be a well-formed code on $n$ neurons such that each neuron fires at least once.

$( \Rightarrow)$ We will proceed by contradiction. Assume that $I_{\mathcal C} = \langle 0 \rangle$ and that $\D_{\mathcal C}$ is not inductively $0$-pierced. By Proposition \ref{prop:crossing}, since $\D_{\mathcal C}$ is not inductively $0$-pierced, there is some well-formed realization of $\D_{\mathcal C}$ with two intersecting curves.  Thus, by Lemma \ref{lem:generator}, the ideal $I_{\mathcal C}$ is non-empty, a contradiction.




$( \Leftarrow)$ Assume that $\mathcal{C}$ is $0$-inductively pierced.  To prove that $I_{\mathcal{C}}=\langle 0 \rangle$, we will proceed by induction on the number of neurons $n$.  If $n=1$, then immediately $I_{\mathcal{C}}=\langle 0 \rangle$.  For the induction step, assume that $I_{\mathcal{C}}=\langle 0 \rangle$ for $n=r$.  Now, let $n=r+1$.  Since $\mathcal{C}$ is $0$-inductively pierced, $\mathcal{C}$ has a $0$-piercing, $\lambda$, such that $\mathcal{C} - \lambda$ is $0$-inductively pierced. Thus, by our induction hypothesis, $I_{\mathcal{C} - \lambda}=\langle 0 \rangle$.

Recall from Section \ref{sec:background-toric-ideals} that we can understand the toric ideal of a code $\mathcal C$, by understanding its hypergraph.  So now let us compare the hypergraphs $\mathcal{H}_{\mathcal{C} - \lambda}$ and $\mathcal{H}_{\mathcal{C}}$. The hypergraph $\mathcal{H}_{\mathcal{C}}$ is obtained from $\mathcal{H}_{\mathcal{C} - \lambda}$ by adding the vertex $i$ and a single edge containing $i$ (only a single edge is added by Lemma \ref{lem:onecodeword}). Since there is only one edge that contains vertex $\lambda$, any primitive balanced edge set must be only on the nodes $[n]\setminus \lambda$, and hence corresponding to a primitive binomial in $I_{\mathcal{C} - \lambda}$. Therefore, since $I_{\mathcal{C} - \lambda}=\langle 0 \rangle$, the ideal $I_{\mathcal{C}}=\langle 0 \rangle.$
\end{proof}

Obtaining results for 1-inductively pierced codes is less direct, and we rely on the machinery in \cite{GP13} to give us a combinatorial perspective on the binomials in the toric ideals.  Theorem \ref{thm:1piercings} states that the toric ideals of 1-inductively pierced codes are generated by quadratics.  Thus, these codes give an infinite family of toric ideals generated by quadratics, which we do not believe have been studied before as class of ideals.

We introduce several lemmas before stating and proving Theorem \ref{thm:1piercings}.  We then end the section with a computational result.  For a subset $W \subseteq [n]$, let $U_W := \cap_{i \in W} U_i$.

\begin{lemma}\label{lem:containment}  Let $\mathcal C \subseteq \{0,1\}^n$ be well-formed and 1-inductively pierced with associated hypergraph $\mathcal H_{\calc}=([n], E(H_{\calc}))$. Let $W \cup \{\lambda\} \in E(H_{\calc})$ with $\lambda \notin W$ and $W \neq \emptyset$. If $U_W\setminus \cup_{i \notin W \cup \{\lambda\}} U_i \subset U_{\lambda}$, then there exists $\lambda_0 \in w$ such that $U_{\lambda_0} \subset U_{\lambda}$, or, in other words, there exists $\lambda_0 \in W$ such that for all $c \in \C$, if $\lambda_0 \in \supp c$  then $\lambda \in \supp c$.
\end{lemma}

\begin{proof}  We will proceed by induction on $n$.  For the base case, let $n=2$ and $W \cup \{\lambda\} \in E(H_{\calc})$.  Then $W$ contains exactly one label, and the statement follows immediately.

Now assume the statement holds when $n \leq r$.  Let $n=r+1$.  Let $W \cup \{\lambda\} \in E(H_{\calc})$ with $\lambda \notin W$ and $W \neq \emptyset$, and assume $U_W\setminus \cup_{i \notin W \cup \{\lambda\}} U_i \subset U_{\lambda}$. Since $\mathcal C$ is inductively pierced, there exists a $0$ or $1$ piercing $\lambda_1$ in $\D_{\mathcal C}$ such that $\mathcal C - \lambda_1$ is 1-inductively pierced. We now investigate four possible cases based on the relationship of $\lambda_1$ to $\lambda$ and $W$.

(\emph{Case 1: $\lambda_1 \notin W$ and $\lambda_1 \neq \lambda$})  Notice that in this case, since $\lambda_1$ is a $0$ or $1$-piercing and can only intersect one other curve, we have $U_W\setminus \cup_{i \notin W \cup \{\lambda\} \cup \{\lambda_1\}} U_i \subset U_{\lambda}$ for any well-formed diagram.  Indeed, if $U_W\setminus \cup_{i \notin W \cup \{\lambda\} \cup \{\lambda_1\}} U_i$ is not a subset of $U_{\lambda}$, then $\lambda_1$ intersects the boundary of $U_W$ at point $q$ outside of $U_{\lambda}$ such that, within an $\epsilon$ ball of $q$, there exists a point in  $U_W \setminus \cup_{i \notin W \cup \{\lambda\} } U_i$ that is also outside of $U_{\lambda}$. Thus, the induction hypothesis is met for $\mathcal C - \lambda_1$, and the statement follows.

(\emph{Case 2: $\lambda_1 = \lambda$}) Notice that a $1$-piercing or $0$ piercing cannot fully contain another curve.  Thus, since $W \neq \emptyset$, and none of the curves whose labels are in $W$ can be fully contained in $U_{\lambda}$, the curve $\lambda_1 = \lambda$ must be a $1$-piercing, and furthermore, all curves whose labels are in $W$ must intersect $\lambda_1=\lambda$.  However, since $\lambda_1=\lambda$ is a $1$-piercing there can only be one curve label $\lambda_2$ in $W$.  Now, since $U_W\setminus \cup_{i \notin W \cup \{\lambda\}} U_i \subset U_{\lambda}$, the set $U_{\lambda_2} \setminus U_{\lambda_1=\lambda}$ must be covered by other fields in any well-formed realization of $\C$, however, by well-formedness, this would imply that $\lambda_1 = \lambda$ intersects at least one more curve in addition to $\lambda_2$, a contradiction.

(\emph{Case 3: $\lambda_1 \in W$ and $\lambda_1$ intersects $\lambda$}) In this case, since $\lambda_1$ cannot intersect any additional curves other than $\lambda$, we have $W=\{ \lambda_1\}$, a contradiction with reasoning parallel to Case 2.

(\emph{Case 4: $\lambda_1 \in W$ and $\lambda_1$ does not intersect $\lambda$})  Since $\lambda_1$ does not intersect $\lambda$, either $U_{\lambda} \subset U_{\lambda_1}$  or  $U_{\lambda_1} \subset U_{\lambda}$.  However, the fact that $\lambda_1$ is a $0$ or $1$-piercing in $\D_{\mathcal C}$ precludes the former, therefore $U_{\lambda_1} \subset U_{\lambda}$.

\end{proof}

\begin{lemma}\label{lem:balancededgesets}  Let $\mathcal C \subseteq \{0,1\}^n$ be well-formed and 1-inductively pierced.  Let $\mathcal H_{\mathcal C}$ be the hypergraph corresponding to $\mathcal C$, and let $\mathcal E=(B, R)$ be a primitive balanced edge set of $\mathcal H_{\mathcal C}$.  Let $Z \subset [n]$ and $\lambda \in [n] \setminus Z$. If $ Z \in B$ and $\{\lambda\} \cup Z \in R$, then there exists a hyperedge $\{\lambda\} \cup W \in B$ such that $U_W\setminus \cup_{i \notin W \cup \{\lambda\}} U_i \nsubseteq U_{\lambda}$ in all realizations of $\mathcal C$.
\end{lemma}

\begin{proof}  Assume the hypotheses, and, proceeding by contradiction, assume for every hyperedge  of the form $\{\lambda\} \cup W \in B$ the relationship $U_W\setminus \cup_{i \notin W \cup \{\lambda\}} U_i \subseteq U_{\lambda}$ holds in every realization of $\mathcal C$, or equivalently, $U_W\setminus \cup_{i \notin W \cup \{\lambda\}} U_i \subset U_{\lambda}$ holds in every well-formed realization of $\mathcal C$.  Let $B_1, \ldots B_t$ be the hyperedges in $B$ that contain $\{ \lambda\}$; note that since $\deg_R \lambda \geq 1$ and $\mathcal E$ is balanced, $B$ contains at least one edge containing $\lambda$.  

By Lemma \ref{lem:containment}, each hyperedge $B_i$  for $1 \leq i \leq t$ contains a vertex that does does not appear in any edge of $H_{\mathcal C}$ without $\lambda$.  Let us denote these vertices uniquely as $\mu_1, \ldots, \mu_s$.  Since $Z$ is a hyperedge of $H_{\mathcal C}$ and does not contain $\lambda$, it must be the case $\mu_i \notin Z$ for $1 \leq i \leq s$.  So combining this with the fact that $\deg_B \mu_i = \deg_R \mu_i$, we have the following degree counts:
$$\deg_B \lambda = \sum_{i=1}^s \deg_B \mu_i, \text{ and}$$
$$\deg_R \lambda \geq \sum_{i=1}^s \deg_R \mu_i + 1=\sum_{i=1}^s \deg_B \mu_i + 1=\deg_B \lambda+1.$$
However, this contradicts $\deg_B \lambda = \deg_R \lambda$.
\end{proof}

\begin{thm}\label{thm:1piercings}  Let $\mathcal C$ be well-formed.  If $\mathcal C$ is $1$-inductively pierced then the toric ideal $I_{\mathcal C}$ is generated by quadratics or $I_{\mathcal C}=\langle 0 \rangle$.
\end{thm}

\begin{proof}
Let $\mathcal C$ be well-formed and assume that $\mathcal C$ is 1-inductively pierced.  Since Theorem \ref{thm:0idealiff0pierced} holds, we can assume without loss of generality that $\mathcal C$ is not 0-inductively pierced. We will now proceed by induction on the number of neurons $n$.  The statement holds for $n=2$, in this case, there is a single code $\mathcal C_2=\{ 00, 10, 01, 11\}$ that is 1-inductively pierced, but not 0-inductively pierced; the toric ideal of $\mathcal C_2$ is generated by a single quadratic, $p_{11} - p_{10}p_{01}$.

For the induction step assume that $I_{\mathcal{C}}$ is generated by quadratics for $n=r$ and let $n=r+1$.  Let $\lambda$ be a 0 or 1-piercing in $\D_{\mathcal C}$.  If $\lambda$ is a 0-piercing then the statement follows by the same argument as in the proof of Theroem \ref{thm:0idealiff0pierced}, so assume that $\lambda$ is a 1-piercing.  In particular, assume that $\lambda$ is a 1-piercing of $\lambda_1$ in $\D_{\mathcal C}$ identified by zone $Z$. Let $f_{\mathcal E}$ be a primitive binomial in $I_{\mathcal C}$ of degree at least three and let $\mathcal E=(B, R)$ be the primitive balanced edge set of $\mathcal H_{\mathcal C}$ that corresponds to $f_{\mathcal E}$.  Without loss of generality, assume that $|B| \geq |R|$.  Note that this means $|B| \geq 3$. 

If no edge in $\mathcal E$ contains the vertex $\lambda$ then $\mathcal E$ is a balanced edge set on $\mathcal H_{\mathcal C - \lambda}$ and, by the induction hypothesis, $f_{\mathcal E}$ is generated by quadratics.  Thus, assume there are edges $B_1 \in B$ and $R_1 \in R$ that contain $\lambda$, furthermore since $\mathcal E$ is primitive, $B_1 \neq R_1$ and thus without loss of generality we can assume $B_1 = \{ \lambda\} \cup Z$ and $R_1 = \{ \lambda, \lambda_1\} \cup Z$, the only two hyperedges in $\mathcal H_{\C}$ that contain $\lambda$.  We will now proceed to find a proper splitting set of $\mathcal E$.  

By Lemma \ref{lem:balancededgesets}, there exists a hyperedge $B_2=\{ \lambda_1 \} \cup W \in B$ such that $U_W\setminus \cup_{i \notin W\cup \{\lambda_1\}} U_i \nsubseteq U_{\lambda_1}$ in all well-formed realizations of $\mathcal{C}$.  Furthermore, by primitivity, $R_2 \neq B_1$ and thus $\lambda \notin W$. Let $d$ be a well-formed realization of $\D_{\mathcal C}$. Since $U_W\setminus \cup_{i \notin W \cup \{\lambda_1\}} U_i \nsubseteq U_{\lambda_1}$ but the intersection of $U_W\setminus \cup_{i \notin W \cup \{\lambda_1\}}U_i$ and $U_{\lambda_1}$ is non-empty, some segment $s$ of the curve $\lambda_1$ is contained in $U_W\setminus \cup_{i \notin W \cup \{\lambda_1\}}U_i$.  Furthermore, since $d$ is well-formed there is an $\epsilon$-neighborhood of $s$ that is fully contained in $U_W\setminus \cup_{i \notin w \cup \{\lambda_1\}}U_i$.  Thus, $U_W\setminus \cup_{i \notin W}U_i$ is non-empty, implying $W$ is an edge in $H_{\mathcal C}$.

Now let $\Gamma_1$ and $\Gamma_2$ be balanced edge sets such that $\Gamma_{1_{blue}}=\{B_1, B_2\}$, $\Gamma_{1_{red}}=\{R_1, W\}$, $\Gamma_{2_{blue}}=B\setminus\{B_1, B_2\} \cup \{W\}$, and $\Gamma_{2_{red}}=R\setminus\{R_1\}$.  Then $S=\{ W \}$ is a proper splitting set of $\mathcal E$ with decomposition $(\Gamma_1, S, \Gamma_2)$ and the proposition follows from Theorem \ref{thm:combinatorialdegreebound}.

\end{proof}

One might expect that the converse of Theorem \ref{thm:1piercings} to be true, however there exists a counterexample with as few as $3$ neurons.

\begin{ex}\label{ex:A1counterexample}
  
The code $\C$=A1=$\{0,1\}^3$ is 2-inductively pierced but not 1-inductively pierced.
Notice that in Example \ref{ex:A1}, the cubic $p_{111}-p_{100}p_{010}p_{001}$ was given as a generator of $I_\C$.  However, this cubic can be written in terms of quadratics.  In particular, $$p_{111}-p_{100}p_{010}p_{001}\ =\ (p_{111}-p_{110}p_{001})+\ p_{001}(p_{110}-p_{100}p_{010}).$$  
Thus, we can give a generating set of the toric ideal of A1 that is generated only by quadratics: $$I_{\mathrm{A1}}=\langle p_{110}-p_{100}p_{010},\ p_{101}-p_{100}p_{001},\ p_{011}-p_{010}p_{001},\ p_{111}-p_{110}p_{001} \rangle.$$  
\end{ex}

While the preceding example shows that the converse of Theorem \ref{thm:1piercings} is false, we do note that two-piercings result in signature cubic binomials in the toric ideal.



\begin{prop}\label{prop:cubic}
Let $\C$ be a well-formed neural code on $n$ neurons.  If there is a triple intersection in a well-formed realization of $\C$ then the toric ideal $I_\C$ contains a cubic binomial, in particular, a binomial of the form $p_{111w} p_{000v}^2-p_{100v}p_{010v}p_{001w}$ or $p_{111w}-p_{1000 \cdots 0}p_{0100 \cdots 0}p_{001w}$ where $v,w \in \{0,1\}^{n-3}$.
\end{prop}

\begin{proof}
Let $\C$ be a well-formed neural code on $n$ neurons with a triple intersection and let $d$ be a well-formed realization of $\C$.   Let us denote the three intersecting curves as $\lambda_1, \lambda_2,$ and $\lambda_3$. Since $\C$ is well-formed, all curves intersect generally, so in particular $\lambda_1$ and $\lambda_2$ intersect at two points, say  $p$ and $q$. We zoom in on the triple intersection, and we have a place field arrangement of $\C$ as illustrated in Figure \ref{fig:tripleintersection}. 


\begin{figure}[!ht] 
\begin{subfigure}[b]{.5\linewidth}
\begin{center}
\begin{tikzpicture}
\draw [name path=1] (-3.75,0.5) arc (-120:90:3cm) node[anchor=east] {$\lambda_1$}; 
\draw [name path=2] (1.25,6) arc (90:300:3cm) node[anchor=west] {$\lambda_2$};
\draw [name path=3] (2.5,-1) arc (0:180:3cm) node[anchor=north east] {$\lambda_3$};

\path [name intersections={of = 1 and 2}];
	\coordinate[label=below:$q$] (q) at (intersection-1);
	\coordinate[label=above:$p$] (p) at (intersection-2);
\fill[fill=red,inner sep=1pt,name intersections={of=1 and 2}]
    (intersection-1) circle (2pt)
    (intersection-2) circle (2pt) ;
\draw[red,thick,dashed] (p) circle (1cm);
\draw[red,thick,dashed] (q) circle (1cm);
\end{tikzpicture}
\caption{A closeup on a triple intersection.}
\label{fig:tripleintersection}
\end{center}
\end{subfigure}

  \begin{subfigure}[b]{0.5\linewidth}
    \centering
\begin{tikzpicture}
\draw [name path=1] (-1,0) arc (0:90:3cm) node[anchor=east] {$\lambda_1$}; 
\draw [name path=2] (0,3) arc (90:180:3cm);
\node at (0.3,3) {$\lambda_2$};
\node at (-3.5,2) {$100v$};
\node at (-0.75,2) {$010v$};
\node at (-2,1) {$110v$};
\node at (-2,3.25) {$000v$};

\path [name intersections={of = 1 and 2,by=p}];

\node [fill=red,inner sep=2pt,label=-90:$p$] at (p) {};
\end{tikzpicture}

    \caption{A closeup on the point of intersection $p$.} 
    \label{fig:pointp} 
    
  \end{subfigure}
  \begin{subfigure}[b]{0.5\linewidth}
    \centering
\begin{tikzpicture}
\draw [name path=1] (-3.75,1) arc (-120:0:3.25cm) node[anchor=east] {$\lambda_1$}; 
\draw [name path=2] (-2.5,3.75) arc (180:300:3.25cm) node[anchor=west] {$\lambda_2$};
\draw [name path=3] (1.75,0) arc (0:180:2.5cm) node[anchor=north east] {$\lambda_3$};

\node at (-2,1) {$101w$};
\node at (.5,1) {$011w$};
\node at (-0.75,1.75) {$111w$};
\node at (-0.75,0) {$001w$};
\path [name intersections={of = 1 and 2,by=q}];
\node [fill=red,inner sep=2pt,label=-90:$q$] at (q) {};
\end{tikzpicture}
    \caption{A closeup on the point of intersection $q$} 
    \label{fig:pointq} 
  \end{subfigure} 
\caption{Closeups on the points of intersection in a triple intersection of curves.}
\label{fig:tripleintersectionfull}
\end{figure}
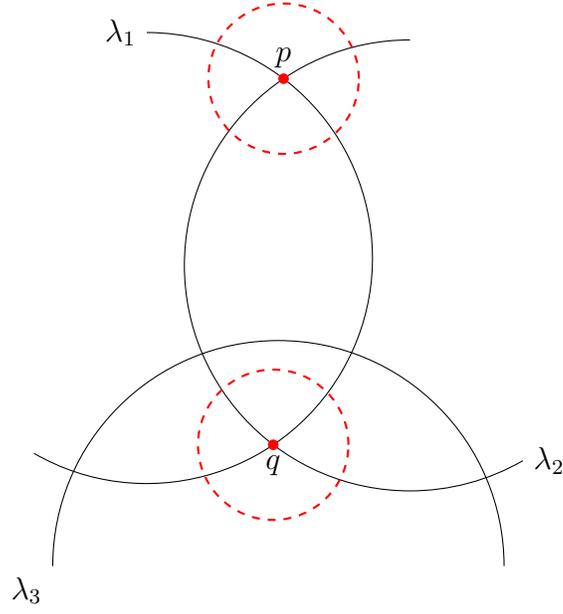
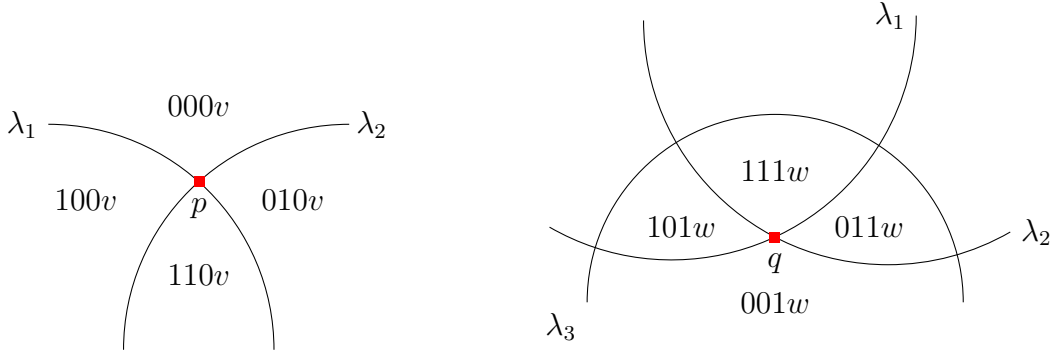

Since $\C$ is well-formed, there exists an open ball around $p$, so the codewords $000v$, $100v$, $010v$, $110v$ must be in $\C$ for some $v\in \{0,1\}^{n-3}$  (see Figure \ref{fig:pointp}).  
Similarly, there exists an open ball around $q$, so the codewords $111w,\ 101w,\ 011w,\ 001w$ must be in $\C$ , for some $w\in \{0,1\}^{n-3}$ (see Figure \ref{fig:pointq}).  

Then the matrix of the codewords in $\C$ looks like 
\[ \left(
\begin{array}{ccccccccc}
0 & 1 & 0 & 1 & 1 & 1 & 0 & 0 &\ldots \\
0 & 0 & 1 & 1 & 1 & 0 & 1 & 0 &\ldots \\
0 & 0 & 0 & 0 & 1 & 1 & 1 & 1 &\ldots \\
v^{\text{T}} & v^{\text{T}} & v^{\text{T}} & v^{\text{T}} & w^{\text{T}} & w^{\text{T}} & w^{\text{T}} & w^{\text{T}} & \ddots
\end{array}
\right), \] 
and, if $v$ is not the all zeros codeword, we have 

\begin{align*}
& \phi_{\C}(p_{000v}) = \prod\limits_{i\in \supp(v)}^{}x_i
& \phi_{\C}(p_{100v}) =  x_1\prod\limits_{i\in \supp(v)}^{}x_i \\
& \phi_{\C}(p_{010v}) =  x_2\prod\limits_{i\in \supp(v)}^{}x_i  
 & \phi_{\C}(p_{001w}) =  x_3\prod\limits_{i\in \supp(w)}^{}x_i \\
& \phi_{\C}(p_{111w}) = x_1x_2x_3\prod\limits_{i\in \supp(w)}^{}x_i.
\end{align*}

Therefore, $p_{111w} p_{000v}^2-p_{100v}p_{010v}p_{001w}$ is in the toric ideal $I_{\C}$.  If $v$ is the all zeros codeword then $p_{111w}-p_{1000 \cdots 0}p_{0100 \cdots 0}p_{001w} \in I_{\C}$.
\end{proof}

In essence, Proposition \ref{prop:cubic} tells us that if we notice three fields intersecting as in Figure \ref{fig:A1}, then we can expect a particular cubic in the toric ideal.  While Example \ref{ex:A1counterexample} shows that it is possible for this cubic to be generated by quadratics in the ideal, we wonder whether there exist term orders such that these signature cubics appear in the reduced Gr\"{o}bner basis.  Using the Macaulay2 interface for {\tt gfan} \cite{gfan}, we are able to find a term order that works in this sense for $n=3$.

In the following proposition, we use a \emph{weighted graded reverse lexicographic } monomial order: Let $\mathbb K [x_1, \ldots, x_n]$ be a polynomial ring, and let $w \in \R^n$ be a weight vector.  Let $deg_w(x^a)=a_1w_1+a_2w_2+\ldots+a_nw_n$ and let $x^a < x^b$ if and only if $deg_w(x^a) < deg_w(x^b)$ or $deg_w(x^a) = deg_w(x^b)$ and there exists $1\leq i \leq n$ such that $a_n=b_n,\ \ldots\ ,a_{i+1}=b_{i+1}$ , $a_i > b_i$ \cite{sage}.  Furthermore, to make the statement of Proposition \ref{prop:3neuron_termorder_GB_gensdeg2} cleaner, we will view each $I_{\C}$ as a subset of the larger polynomial ring $\mathbb K[ p_{100}, p_{010}, p_{001}, p_{110}, p_{101}, p_{011}, p_{111}]$.

\begin{prop}\label{prop:3neuron_termorder_GB_gensdeg2}
A well-formed neural code $\C$ on 3 neurons is 1-inductively pierced if and only if the Gr\"{o}bner basis of $I_{\C}$ with respect to the weighted graded reverse lexicographic order with the weight vector $w=(0,0,0,1,1,1,0)$ contains only binomials of degree $2$ or less.
\end{prop}

\begin{proof}
Using the weighted graded reverse lexicographic order with weight vector $w=(0,0,0,1,1,1,0)$ we computed the reduced Gr\"{o}bner bases of the toric ideals of each well-formed neural code up to symmetry. We found that only the 0- and 1- inductively pierced codes had reduced Gr\"{o}bner bases with maximum degree two.  
\end{proof}
We end this section with the following conjecture.

\begin{conj}
For each $n$, there exists a term order such that a code is 0- or 1-inductively pierced if and only if the reduced Gr\"{o}bner basis contains binomials of degree 2 or less.
\end{conj}

\section{Conclusion: Drawing place fields for neural codes}\label{sect:conclusion}
Our motivating question is how to draw the realization of a place field diagram for a neural code assuming we know \emph{a priori} that it is convexly realizable in dimension two.  Existing work in the field of information theory \cite{SZHR11} gives an algorithm for drawing such realizations when the data sets, i.e. codes, are inductively pierced.  Thus, our question of focus for this manuscript is how to determine whether a neural code is inductively pierced, or $k$-inductively pierced.  
 To this end, we utilized two different algebraic objects, the neural ideal and the toric ideal of a neural code.  The following theorem summarizes our results.

\begin{thm}\label{thm:mainthm} (Summary of results)
Let $\C$ be a well-formed neural code on $n$ neurons.
\begin{enumerate}
\item The neural code $\C$ is $0$-inductively pierced if and only if $CF(J_\C) = \{f_{ij} \,|\, i,j\in [n], i<j, f_{ij}  \in \{x_ix_j, x_i(1-x_j), x_j(1-x_i) \}$.  
\item $\C$ is 1-inductively pierced if and only if $CF(J_\C)$ consists only of degree two pseudo-monomials meeting the following conditions:
\begin{enumerate}[(i)]
\item For each pair $\{i,j\}$, at most one of $x_ix_j, x_i(1-x_j), x_j(1-x_i)$ appears in $CF(J_\C)$.
\item $G(\C)$ is a forest.
\end{enumerate}
\item The neural code $\C$ is 0-inductively pierced if and only if $I_{\C}= \langle 0 \rangle$.
\item If the neural code $\C$ is 0 and 1-inductively pierced then $I_{\C}$ is $\langle 0 \rangle$ or generated by quadratics.

\end{enumerate}
\end{thm}

Using the canonical form of the neural ideal, we have fully classified 0- and 1-inductively pierced codes. For toric ideals, we have a full understanding of $0$-inductively pierced codes and a partial understanding of $1$-inductively pierced codes. However, in the big picture, this work is still in progress. One goal for further work is to completely classify $k$-inductively pierced codes using their toric ideals; another is classify them using their canonical forms. 



We end here with a large example, illustrating our results. Consider the following neural code on 17 neurons:
\begin{align*}
\C=\{&c_0=00000000000000000, 
c_1=10000000000000000,
c_2=11000000000000000,\\
&c_3=11100000000000000,
c_4=10100000000000000,
c_5=10010000000000000,\\
&c_6=10011000000000000,
c_7=00010000000000000,
c_8=00011000000000000,\\
&c_9=00011100000000000,
c_{10}=00010100000000000,
c_{11}=00010010000000000,\\
&c_{12}=00010001000000000,
c_{13}=00010001100000000,
c_{14}=00000001100000000,\\
&c_{15}=00000001110000000,
c_{16}=00000001010000000,
c_{17}=00000001011000000,\\
&c_{18}=00000001001000000,
c_{19}=00000001000000000,
c_{20}=00000001000100000,\\
&c_{21}=00000001000110000,
c_{22}=00000001000010000,
c_{23}=00000001000011000,\\
&c_{24}=00000000000010000,
c_{25}=00000000000011000,
c_{26}=00000000000010100,\\
&c_{27}=00000000000010010,
c_{28}=00000000000010001\}.
\end{align*}  
We compute the canonical form of the neural ideal $CF(J_\C$), determine the graph of the code $G(\C)$, and compute its toric ideal $I_\C$.  The toric ideal of $\C$ is
\\

$I_\C= \langle p_{c_{20}}p_{c_{24}}-p_{c_{21}},\ p_{c_{19}}p_{c_{24}}-p_{c_{22}},\ p_{c_{19}}p_{c_{25}}-p_{c_{23}},\ p_{c_{16}}p_{c_{18}}-p_{c_{23}}, 
\ p_{c_{16}}p_{c_{18}}-p_{c_{17}}p_{c_{19}},\\ p_{c_{14}}p_{c_{16}}-p_{c_{15}}p_{c_{19}},\ p_{c_{7}}p_{c_{9}}-p_{c_{8}}p_{c_{10}},\ p_{c_{7}}p_{c_{19}}-p_{c_{12}},\ p_{c_{7}}p_{c_{14}}-p_{c_{13}},\ p_{c_{1}}p_{c_{3}}-p_{c_{2}}p_{c_{4}},\\ p_{c_{1}}p_{c_{7}}-p_{c_{5}},\ p_{c_{1}}p_{c_{8}}-p_{c_6} \rangle.$\\

From these computations we see that $I_\C$ is generated by binomials of degree at most 2. Thus, from our results we know that $\C$ is not $0$-inductively pierced and is possibly $1$-inductively pierced.  The canonical form will give us more information.


The pseudo-monomials in the canonical form and the graph $G(\C)$ are listed in Appendix \ref{app:17neuroncomputation}. Since for each pair $\{i,j\}$, at most one of $x_ix_j$, $x_i(1-x_j)$, and $x_j(1-x_i)$ appears in $CF(J_\C)$ and since $G(\C)$ is a forest, by Theorem \ref{thm:01piercing}, the code $\C$ is 1-inductively pierced.  Thus, we can use the existing algorithm in \cite{SZHR11} that draws Euler diagrams with circles. The algorithm is implemented and available at \url{http://www.eulerdiagrams.org/inductivecircles.html}.  Figure \ref{fig:17neuroncodeexample} shows the input and output of the program. Note that to input the code in this program we rename each codeword to its support, where a=1, b=2, etc., omitting commas and braces.  The output of the program is a place field diagram of $\C$.  

Finally, from Algorithm \ref{alg:ordering} in Section \ref{sect:canonicalform}, we determine a drawing order for this place field diagram of $\C$ as follows: 
\begin{enumerate}[(i)]

\begin{minipage}{0.6\textwidth}
\item 1 (0-piercing in $\emptyset$), 
\item 4 (1-piercing of 1), 
\item 8 (1-piercing of 4), 
\item 13 (1-piercing of 8),
\item 5 (1-piercing of 1 in 4),
\item 6 (1-piercing of 5 in 4),
\item 9 (1-piercing of 4 in 8),
\item 10 (1-piercing of 9 in 8),
\end{minipage}
\begin{minipage}{0.55\textwidth}
\item 11 (1-piercing of 10 in 8),
\item 12 (1-piercing of 13 in 8),
\item 14 (1-piercing of 8 in 13),
\item 2 (0-piercing in 1),
\item 3 (1-piercing of 2 in 1),
\item 7 (0-piercing in 4),
\item 15 (0-piercing in 13),
\item 16 (0-piercing in 13),
\item 17 (0-piercing in 13). 
\end{minipage}
\end{enumerate}

\begin{figure}[!ht]
\begin{center}
\includegraphics[width=4in]{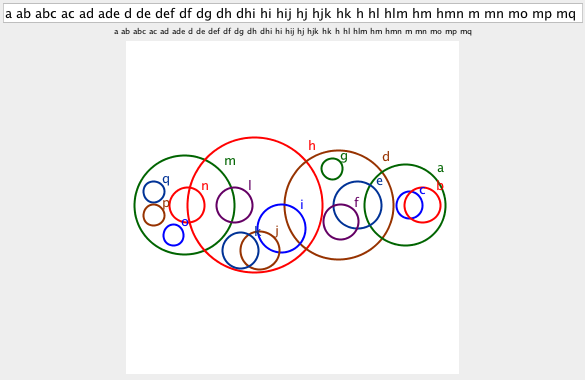}
\caption{A place field diagram of a 17-neuron code drawn using the implemented algorithm from \cite{SZHR11}. }
\label{fig:17neuroncodeexample}
\end{center}
\end{figure}

\section{Acknowledgements}  This collaboration is a result of the 2014 AMS Mathematics Research Community, ``Algebraic and Geometric Methods in Applied Discrete Mathematics," which was supported by NSF DMS-1321794. Elizabeth Gross was supported by NSF DMS-1304167.  The authors would like to thank Tim Hsu and Richard Kulbeka for their extensive comments on a preliminary version of this work. 

\newpage
\appendix
\section{\\Generators for the $n=3$ case} \label{App:AppendixA}

Below is a table of the generating sets of $I_A$, the toric ideal, for the different codes on $n=3$ neurons listed in Figure 6 of the original paper on the neural ring \cite{CICY2013}.   \\

\begin{longtable}{|c|c|}
	\hline {\bf Generators of $I_A$} & {\bf Codes}  \\ \hline
	$p_{111}-p_{100}p_{010}p_{001}$	& 	$A1$	 \\  
	$p_{110}-p_{100}p_{010}$ & \\
	$p_{101}-p_{100}p_{001}$ & \\
	$p_{011}-p_{010}p_{001}$ & \\ \hline
	
	$p_{111}-p_{010}p_{101}$ & $A2$ \\
	$p_{110}-p_{100}p_{010}$ & \\ \hline
	
	$p_{111}-p_{100}p_{010}p_{001}$ & $A3$ \\
	$p_{110} - p_{100}p_{010}$ & \\
	$p_{101} - p_{100}p_{001}$ & \\ \hline
	
	$p_{111}-p_{100}p_{011}$ & $A4$ \\
	$p_{110} - p_{100}p_{010}$ & \\
	$p_{100}p_{011}- p_{010}p_{101}$ & \\ \hline
	
    $p_{110} - p_{100}p_{010}$ & $ A5, B2, C1, F1 $\\ \hline
	
	$p_{100}p_{111} - p_{110}p_{101}$ & $A6$ \\ \hline
	
	$p_{111}-p_{010}p_{101}$ & $A7$ \\ \hline
	
	$p_{111}-p_{100}p_{010}p_{001}$ & $A8$ \\ 
	$p_{110}-p_{100}p_{010}$ & \\ \hline
	
	$p_{111}-p_{100}p_{011}$ & $A9,A16$ \\ \hline
	
    $p_{111}-p_{010}p_{101}$ & $A10$ \\
    $p_{100}p_{011}-p_{010}p_{101}$ & \\ \hline
    
	$p_{100}^2p_{011}-p_{110}p_{101}$ & $A11$ \\
	$p_{111}-p_{100}p_{011}$ & \\ \hline 
	
    $p_{111}-p_{100}p_{010}p_{001}$ & $A14$ \\ \hline

    $p_{111}^2-p_{110}p_{101}p_{011}$ & $A15$ \\ \hline    
    
    $p_{101}-p_{100}p_{001}$ & $B1$ \\
    $p_{110}-p_{100}p_{010}$ & $ $ \\ \hline
    
    $p_{100}p_{011}-p_{010}p_{101}$ & $B3$ \\ \hline
    
    $p_{011}-p_{010}p_{001}$ & $E1$ \\
    $p_{101}-p_{100}p_{001}$ & $ $ \\ 
    $p_{110}-p_{100}p_{001}$ & $ $ \\ \hline
    
    $p_{110}-p_{100}p_{010}$ & $E2$ \\ 
    $p_{100}p_{011}-p_{010}p_{101}$ & $ $ \\
    $ $ & $ $ \\ \hline

    $p_{100}^2p_{011}-p_{110}p_{101}$ & $E3$ \\ \hline

	$0$ & $A12, A13, A17, A18, A19, A20, B4, B5, B6, C2, C3, D1, E4, F2, F3, G1, H1, I1$ \\ \hline


\end{longtable} 



\newpage
\section{\\Computation for a 17-neuron code}
\label{app:17neuroncomputation}
The generators of the toric ideal of the 17-neuron code from Section \ref{sect:conclusion}:

$I_\C= \langle p_{c_{20}}p_{c_{24}}-p_{c_{21}},\ p_{c_{19}}p_{c_{24}}-p_{c_{22}},\ p_{c_{19}}p_{c_{25}}-p_{c_{23}},\ p_{c_{16}}p_{c_{18}}-p_{c_{23}},\ p_{c_{16}}p_{c_{18}}-p_{c_{17}}p_{c_{19}},\ p_{c_{14}}p_{c_{16}}-p_{c_{15}}p_{c_{19}},\ p_{c_{7}}p_{c_{9}}-p_{c_{8}}p_{c_{10}},\ p_{c_{7}}p_{c_{19}}-p_{c_{12}},\ p_{c_{7}}p_{c_{14}}-p_{c_{13}},\ p_{c_{1}}p_{c_{3}}-p_{c_{2}}p_{c_{4}},\ p_{c_{1}}p_{c_{7}}-p_{c_{5}},\ p_{c_{1}}p_{c_{8}}-p_{c_6} \rangle$.

\begin{table}
\begin{center}
\begin{tabular}{|c|c|c|c|c|}

  \hline
$x_2(1-x_1)$  	&	$x_2x_4$   	&	$x_3x_{15}$ 	&	$x_6x_{10}$	&	$x_9x_{15}$ 	\\ \hline
$x_3(1-x_1)$	&	$x_2x_5$  	&	$x_3x_{16}$ 	&	$x_6x_{11}$	&	$x_9x_16$   	\\ \hline
$x_5(1-x_4)$ 	&	$x_2x_6$   	&	$x_3x_{17}$  	&	$x_6x_{12}$	&	$x_9x_{17}$ 	\\ \hhline{|-|-|=|-|=|}
$x_6(1-x_4)$ 	&	$x_2x_7$ 	&	$x_4x_{10}$	&	$x_6x_{13}$	&	$x_{10}x_{12}$	\\ \hline
$x_7(1-x_4)$  	&	$x_2x_8$	&	$x_4x_{11}$	&	$x_6x_{14}$	&	$x_{10}x_{13}$  	\\ \hline
$x_9(1-x_8)$	&	$x_2x_9$	&	$x_4x_{12}$	&	$x_6x_{15}$	&	$x_{10}x_{14}$ 	\\ \hline
$x_{10}(1-x_8)$ 	&	$x_2x_{10}$	&	$x_4x_{13}$	&	$x_6x_{16}$ 	&	$x_{10}x_{15}$	\\ \hline
$x_{11}(1-x_8)$ 	&	$x_2x_{11}$	&	$x_4x_{14}$ 	&	$x_6x_{17}$ 	&	$x_{10}x_{16}$	\\ \hhline{|-|-|-|=|-|}
$x_{12}(1-x_8)$ 	&	$x_2x_{12}$  	&	$x_4x_{15}$ 	&	$x_7x_8$         	&	$x_{10}x_{17}$       	\\ \hhline{|-|-|-|-|=|}
$x_{14}(1-x_{13})$	&	$x_2x_{13}$	&	$x_4x_{16}$	&	$x_7x_9$           	&	$x_{11}x_{12}$    	\\ \hline
$x_{15}(1-x_{13})$	&	$x_2x_{14}$ 	&	$x_4x_{17}$      	&	$x_7x_{10}$     	&	$x_{11}x_{13}$     	\\ \hhline{|-|-|=|-|-|}
$x_{16}(1-x_{13})$	&	$x_2x_{15}$ 	&	$x_5x_7$ 	&	$x_7x_{11}$ 	&	$x_{11}x_{14}$ 	\\ \hline
$x_{17}(1-x_{13})$  	&	$x_2x_{16}$	&	$x_5x_8$	&	$x_7x_{12}$  	&	$x_{11}x_{15}$	\\ \hhline{|=|-|-|-|-|}
$x_1x_6$	&	$x_2x_{17}$      	&	$x_5x_9$	&	$x_7x_{13}$ 	&	$x_{11}x_{16}$ 	\\ \hhline{|-|=|-|-|-|}
$x_1x_7$	&	$x_3x_4$   	&	$x_5x_{10}$       	&	$x_7x_{14}$ 	&	$x_{11}x_{17}$ 	\\ \hhline{|-|-|-|-|=|}
$x_1x_8$ 	&	$x_3x_5$            	&	$x_5x_{11}$ 	&	$x_7x_{15}$ 	&	$x_{12}x_{14}$      	\\ \hline
$x_1x_9$ 	&	$x_3x_6$          	&	$x_5x_{12}$	&	$x_7x_{16}$	&	$x_{12}x_{15}$	\\ \hline
$x_1x_{10}$ 	&	$x_3x_7$	&	$x_5x_{13}$	&	$x_7x_{17}$	&	$x_{12}x_{16}$	\\ \hhline{|-|-|-|=|-|}
$x_1x_{11}$ 	&	$x_3x_8$	&	$x_5x_{14}$  	&	$x_8x_{15}$ 	&	$x_{12}x_{17}$        	\\ \hhline{|-|-|-|-|=|}t
$x_1x_{12}$	&	$x_3x_9$	&	$x_5x_{15}$	&	$x_8x_{16}$ 	&	$x_{14}x_{15}$	\\ \hline
$x_1x_{13}$ 	&	$x_3x_{10}$	&	$x_5x_{16}$ 	&	$x_8x_{17}$   	&	$x_{14}x_{16}$   	\\ \hhline{|-|-|-|=|-|}
$x_1x_{14}$	&	$x_3x_{11}$	&	$x_5x_{17}$ 	&	$x_9x_{11}$     	&	$x_{14}x_{17}$	\\ \hhline{|-|-|=|-|=|}
$x_1x_{15}$	&	$x_3x_{12}$ 	&	$x_6x_7$         	&	$x_9x_{12}$  	&	$x_{15}x_{16}$       	\\ \hline
$x_1x_{16}$	&	$x_3x_{13}$	&	$x_6x_8$   	&	$x_9x_{13}$  	&	$x_{15}x_{17}$ 	\\ \hhline{|-|-|-|-|=|}
$x_1x_{17}$ 	&	$x_3x_{14}$	&	$x_6x_9$	&	$x_9x_{14}$	&	$x_{16}x_{17}$	\\ \hhline{|=|-|-|-|-|}
\end{tabular}
\caption{Canonical form of the neural ideal for the 17-neuron code from Section \ref{sect:conclusion}.}
\end{center}
\end{table}

\begin{figure}[h]
\includegraphics[width=.5\textwidth]{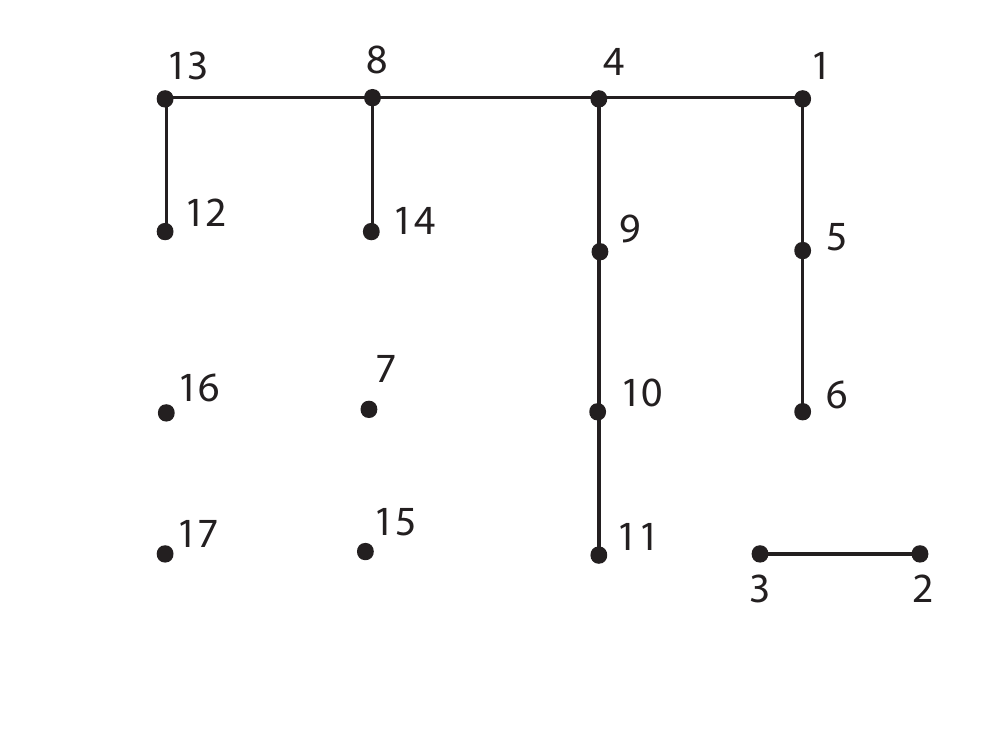}
\caption{The graph $G(C)$ of the 17-neuron code.}
\end{figure}

\end{document}